\documentclass[11pt]{article}

\usepackage[a4paper,margin=1in]{geometry}
\usepackage{setspace}
\usepackage{newtxtext}

\usepackage[numbers,sort&compress]{natbib}

\usepackage{float}
\usepackage{array}

\usepackage{booktabs}
\usepackage[font=small]{caption}

\usepackage{amsmath,amssymb,amsthm}
\numberwithin{equation}{section}
\usepackage{authblk}

\usepackage{xcolor}
\usepackage{graphicx}
\graphicspath{{figures/}}

\usepackage{enumitem}

\usepackage{microtype}

\usepackage{csquotes}

\usepackage{xurl}

\usepackage[
  colorlinks=true,
  linkcolor=blue,
  citecolor=blue,
  urlcolor=blue
]{hyperref}

\usepackage[nameinlink,noabbrev]{cleveref}

\newcommand{\sagenew}[1]{#1}
\newcommand{\sagecolor}{}

\theoremstyle{definition}
\newtheorem{definition}{Definition}[section]
\newtheorem{remark}{Remark}[section]

\theoremstyle{definition}

\newtheorem{proposition}{Proposition}[section]
\newtheorem{theorem}{Theorem}[section]
\newtheorem{corollary}{Corollary}[section]

\newcommand{\Q}{Q}

\title{A Family of Quantile Functions Useful in Clinical Studies}
\author[1]{\protect\sagenew{P.\,G. Sankaran}}
\author[1]{\protect\sagenew{V.\,P. Prasanth}\thanks{Corresponding author: \protect\sagenew{V.\,P. Prasanth, Department of Statistics, Cochin University of Science and Technology, Kochi 682022, Kerala, India. Email: }\url{prasanth.stat@gmail.com}}}
\author[2]{\protect\sagenew{N.\,N. Midhu}}
 
\affil[1]{Department of Statistics, Cochin University of Science and Technology, Kochi \protect\sagenew{682022, Kerala}, India}
\affil[2]{\protect\sagenew{Biostatistics, IQVIA, Prestige TMS Square, Edappally, Kochi 682024, Kerala, India}}
\date{\small \protect\sagenew{ORCID iDs: P.\,G. Sankaran, 0009-0003-0163-8858; V.\,P. Prasanth, 0009-0005-8660-0035; N.\,N. Midhu, 0000-0001-7570-3446}}

\begin{document}
\maketitle

\begin{abstract}
Survival comparisons usually rest on means, medians or hazard ratios. These summaries say little about patients who live beyond a given time, and two treatments with similar medians can differ in this upper tail. We therefore study the quantile-based effectiveness persistence function defined as the ratio between the tail mean and the quantile function. Motivated by upper-tail quantile-domain summaries, we derive statistical properties of this measure and consider a rational (M\"obius) specification of the quantile-based effectiveness persistence function. Under natural boundary conditions, this specification reduces to a canonical form, in which remaining life relative to elapsed time is proportional to the survival odds. The resulting canonical family defines a two-parameter class of non-negative distributions through its quantile function. Various properties, including descriptive measures, L-moments, and quantile-based reliability concepts, are derived for this class. For two treatment arms, the family gives a contrast that is constant over quantile levels without assuming proportional hazards. Estimation of the model parameters using maximum likelihood is also developed. The proposed family is illustrated using two real survival datasets.
\end{abstract}


\noindent\protect\sagenew{\textbf{AMS 2020 subject classification:} Primary 62E10; Secondary 62N05, 62N02, 62P10.}

\makeatletter
\renewenvironment{proof}[1][Proof]{%
  \par\noindent
  \textbf{#1: }\normalfont
}{\hfill$\square$\par}
\makeatother

\noindent
\section{Introduction}
\label{sec:intro}

The mean residual life function is a basic tool for describing the upper tail of a lifetime distribution. When a distribution is specified through its quantile function $Q(u)$, the corresponding measure is the mean residual quantile function $M(u)$ \citep{NairSankaranBalakrishnan2013}. Quantile-based measures have two practical advantages: several useful models have a quantile function in closed form but no closed-form distribution function, and quantile-based summaries are less affected by a few long-term survivors \citep{Gilchrist2000,NairSankaranBalakrishnan2013}. Upper-tail measures are also widely used outside medicine, for example in hydrology and climate studies \citep{ElAdlouni2008,Papalexiou2013,Malamud2004}. \citet{qepf2026} introduced the quantile-based effectiveness persistence function (QEPF), the mean of the outcomes above the $u$-th quantile divided by that quantile, to compare the upper tails of responses in biosimilar studies.

In a clinical follow-up study a natural quantile question is the following: among patients who survive beyond the $u$-th quantile $Q(u)$ of the survival time, how much longer do they live, relative to the time already elapsed? The answer is the relative mean residual quantile function $M(u)/Q(u)$, and since $M(u)/Q(u)=P(u)-1$, where $P(u)$ denotes the QEPF, it is measured by the QEPF. Two boundary conditions constrain $P(u)$. A survival time that starts at zero has $Q(0^+)=0$; as $u\to0$ the mean of the outcomes above $Q(u)$ tends to the mean while $Q(u)\to0$, so $P(0^+)=\infty$. A survival time observed within a bounded follow-up period has a finite upper endpoint; the mean above $Q(u)$ and $Q(u)$ then approach the same limit, so $P(1^-)=1$ (\Cref{rem:P_limit}). A constant QEPF satisfies neither condition, because it characterizes the Pareto distribution, which has $Q(0^+)>0$. The simplest non-constant rational form, a linear-fractional (M\"obius) function of $u$, is therefore the parsimonious choice that can satisfy both conditions. \Cref{fig:empirical_P} shows the behaviour that this choice must reproduce for the 43 blood cancer patients of \citet[Table~4]{Abouammoh1994}: the empirical QEPF of \citet{qepf2026} is large at small $u$ and falls towards $1$ as $u$ increases.

\begin{figure}[ht]
\centering
\includegraphics[width=0.75\textwidth]{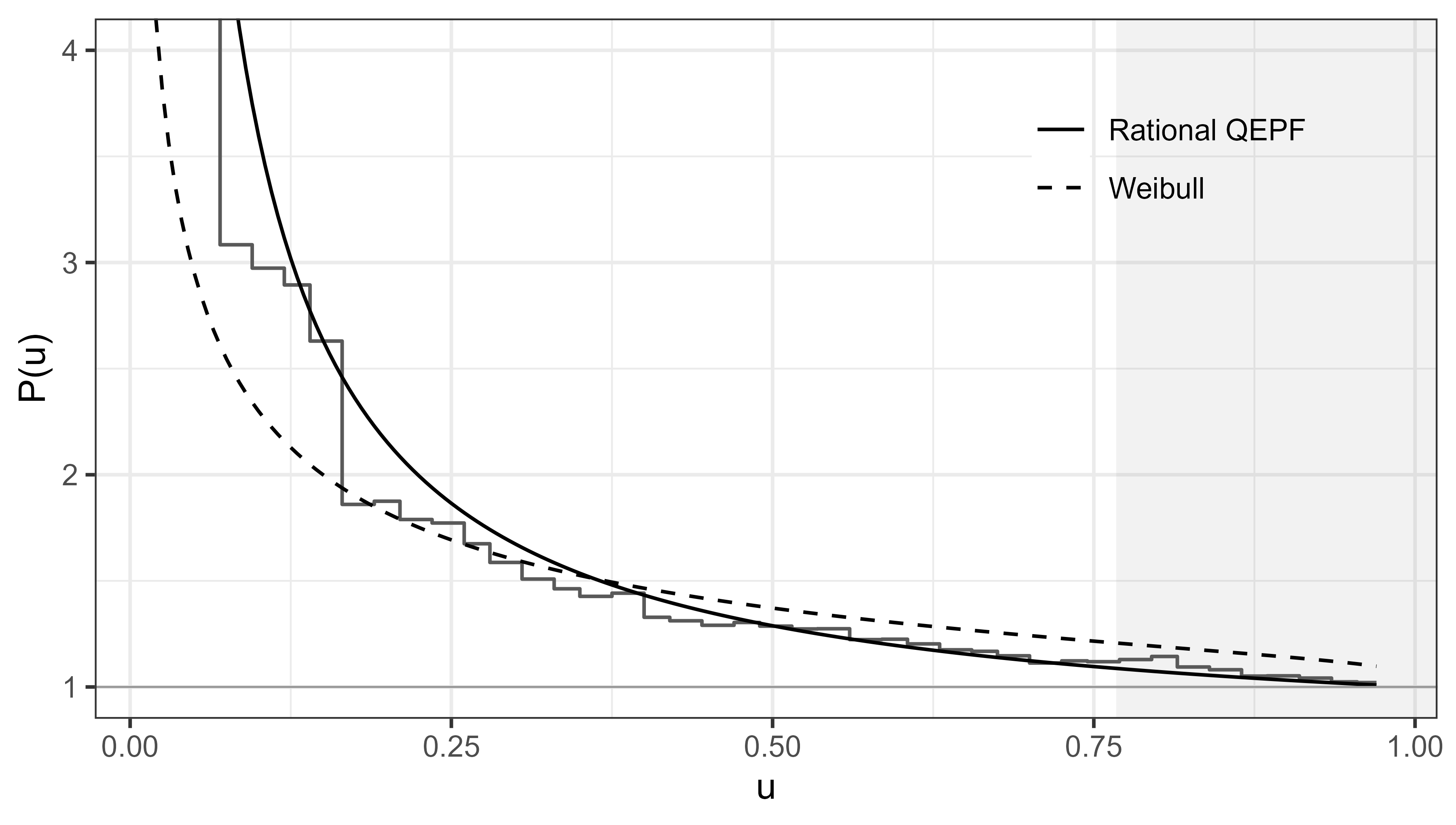}
\caption{Empirical QEPF $\hat P(u)$ (step function) for the blood cancer data of \citet[Table~4]{Abouammoh1994} ($n=43$, complete), with the QEPF of the fitted Rational QEPF (solid) and of the fitted Weibull model (dashed). In the shaded region fewer than $10$ observations lie above $\hat Q(u)$. \protect\sagenew{Source: data from \citet[Table~4]{Abouammoh1994}.}}
\label{fig:empirical_P}
\end{figure}

Imposing the two conditions on the linear-fractional form leaves a single shape parameter $\gamma>-1$ (\Cref{cor:canonical}), and gives $M(u)/Q(u)=c(1-u)/u$ with $c=1/(1+\gamma)$. The ratio $(1-u)/u=\{1-F(Q(u))\}/F(Q(u))$, where $F$ is the distribution function, is the survival odds at $Q(u)$, the quantity on which the proportional odds model is built \citep{KirmaniGupta2001}. Beyond each quantile, therefore, remaining life relative to elapsed time is proportional to the survival odds. Three practical consequences follow. First, the quantile function is available in closed form and has two parameters, the mean $\mu$ and the shape $\gamma$; the family is thus as parsimonious as the other two-parameter lifetime models, such as the Weibull and gamma models. Second, relaxing the upper condition multiplies $P(u)$ by a constant $\lambda=P(1^-)\ge1$, so a likelihood ratio test of $\lambda=1$ tests for a finite endpoint. Because follow-up in a clinical study is always bounded, a finite endpoint describes survival within the follow-up period, not a bounded human lifetime. Third, for a treatment arm $T$ and a control arm $C$, the ratio $\theta=(1+\gamma_C)/(1+\gamma_T)$ of $M(u)/Q(u)$ between the arms is constant over $u$ and compares the arms without assuming proportional hazards.

Several classes of distributions have been defined through quantile-based reliability measures. \citet{MidhuSankaranNair2013} characterized a class through a linear mean residual quantile function. \citet{MidhuSankaran2014QDF} defined models through the quantile density function, and \citet{SankaranMidhu2017MRQRC} estimated the mean residual quantile function nonparametrically under right censoring. \citet{SankaranKumar2018} proposed a family defined by its quantile function and estimated it by L-moments. Distributions defined directly through a quantile function, such as Tukey's g-and-h and the g-and-k families, are also well established \citep{Gilchrist2000,RaynerMacGillivray2002}. Earlier quantile-based families were generated by specifying the mean residual quantile function or the hazard--quantile function \citep{MidhuSankaranNair2013,MidhuSankaranNair2014,SankaranThomasMidhu2015}. Here we instead specify the \emph{relative} mean residual quantile function $M(u)/Q(u)=P(u)-1$. Two recent classes extend the linear model by specifying $M(u)$ directly: the quadratic class $M(u)=A+Bu+Cu^2$ \citep{SankaranDileepKumar2019} and the bilinear class $M(u)=\kappa(1+Au)/(1+Bu)$ \citep{AswinSankaranSunoj2023}. The Rational QEPF is not a special case of either. Its mean residual quantile function, $M(u)=\mu(1-u)(1+\gamma u)^{1/\gamma-1}$, where $\mu>0$ is the mean and $\gamma>-1$ a shape parameter, is bilinear only for $\gamma=1$, the uniform distribution, and quadratic only for $\gamma=1$ and $\gamma=1/2$. The classes therefore share at most two members, and the condition that defines the Rational QEPF, $M(u)/Q(u)=(1-u)/\{(1+\gamma)u\}$, has no counterpart in them.

The contributions of the paper are as follows. First, we show that the QEPF determines a distribution up to scale and characterize the functions that can be a QEPF. Second, we study QEPFs that are rational, specifically linear-fractional, in the quantile level $u$, and show that the two boundary conditions of a lifetime with a finite maximal value, $Q(0^+)=0$ and $Q(1^-)<\infty$, reduce this class to the canonical form
\begin{equation*}
P(u)=\frac{1+\gamma u}{(1+\gamma)u},
\qquad\text{equivalently}\qquad
\frac{M(u)}{Q(u)}=\frac{1}{1+\gamma}\,\frac{1-u}{u},\qquad \gamma>-1.
\end{equation*}
By the second form, beyond the $u$-th quantile, the expected residual lifetime relative to the quantile itself is proportional to the survival odds $(1-u)/u=\{1-F(Q(u))\}/F(Q(u))$. Third, the canonical form defines a two-parameter family of non-negative distributions, which we call the \emph{Rational QEPF} family, with closed-form quantile function, tail mean, mean residual quantile function, hazard--quantile function and L-moments. Fourth, we show that this proportionality characterizes the family: a non-negative random variable with finite mean satisfies $M(u)/Q(u)=c(1-u)/u$ for all $u$ if and only if it belongs to the Rational QEPF family, and $Q(0^+)=0$ and a finite upper endpoint follow from this condition rather than being assumed. Fifth, dropping the condition $Q(1^-)<\infty$ multiplies the Rational QEPF by a constant $\lambda\ge1$ and gives Pareto-type upper tails. This extension removes the restriction to bounded support and turns the question \protect\sagenew{`is the support bounded?'} into a test of $\lambda=1$. Sixth, for two treatment arms the Rational QEPF gives a contrast $\theta=(1+\gamma_C)/(1+\gamma_T)$, the ratio of $M(u)/Q(u)$ between the arms, which does not depend on $u$ and does not assume proportional hazards.

The rest of the article is organized as follows. Section~\ref{sec:prelim} introduces the QEPF and the related quantile-domain measures, and develops the general rational specification. Section~\ref{sec:closed_form} studies the Rational QEPF family, its hazard--quantile function and its extension to unbounded support. A log transformation, which maps the family into the fractional-linear hazard--quantile class, and the L-moments of the family are given in Section~\ref{sec:characterizations}. Statistical inference, including the test of bounded support, is developed in Section~\ref{sec:inference}, a simulation study is reported in Section~\ref{sec:simulation}, and applications to two survival datasets are presented in Section~\ref{sec:data-analysis}. Finally, Section~\ref{sec:conclusion} gives the conclusions.

\section[Preliminaries: QEPF]{Preliminaries}
\label{sec:prelim}

Let \(X\) be a continuous non-negative random variable with distribution function \(F(x)\), density \(f(x)\) and quantile function
\begin{equation}
Q(u)=F^{-1}(u)=\inf\{x:F(x)\ge u\}, \qquad 0<u<1.
\end{equation}
If \(F\) is right-continuous and strictly increasing, then
\begin{equation}
F(Q(u))=u, \qquad 0<u<1.
\end{equation}
Thus the quantile function provides an equivalent representation of the distribution. For general properties and applications of quantile functions, see \cite{Parzen1979, Gilchrist2000} and \cite{NairSankaranBalakrishnan2013, NairSankaranSunoj2019}. In survival studies, the vitality function of $X$ is defined by
\begin{equation}
\label{eq:vitality}
v(x)=\mathbb{E}(X\mid X>x),
\end{equation}
which represents the mean of $X$ beyond the threshold $x$. The quantile counterpart is obtained by evaluating the vitality function at $x=Q(u)$. The quantile analogue of \Cref{eq:vitality} is the \emph{quantile-based vitality function}, defined by \citep{NairSankaranBalakrishnan2013}
\begin{equation}
\label{eq:vitality_def}
V(u)=\frac{1}{1-u}\int_u^1 Q(p)\,dp,
\qquad 0<u<1.
\end{equation}
This is closely related to expected shortfall in risk theory \citep{AcerbiTasche2002}.

Motivated by this idea, we now introduce the quantile-based effectiveness
persistence function and related quantile-domain measures. The purpose of this
section is twofold. First, we collect the basic identities linking the QEPF,
the mean residual quantile function, and the hazard--quantile function.
Second, we develop a general rational specification for the QEPF.
The family studied in \Cref{sec:closed_form} is the subclass of this
specification that is compatible with $Q(0^+)=0$.

\begin{definition}[]
The \emph{quantile-based effectiveness persistence function (QEPF), $P(u)$} is defined as the ratio of the tail mean to the quantile function: \citep{qepf2026}
\begin{equation}
\label{eq:qepf_def}
P(u)=\frac{V(u)}{Q(u)}
=\frac{1}{(1-u)Q(u)}\int_u^1 Q(p)\,dp,
\qquad 0<u<1.
\end{equation}

\noindent
$P(u)$ is the mean of the top $100(1-u)\%$ of outcomes divided by the threshold $Q(u)$. Since $Q$ is increasing, $V(u)>Q(u)$ and hence $P(u)>1$ on $(0,1)$. $P$ is unchanged when $X$ is replaced by $cX$, $c>0$.

\end{definition}

In reliability theory, the mean residual life function of a non-negative random variable \(X\) is defined by

\begin{equation*}
m(x)=\mathbb{E}(X-x\mid X>x),
\end{equation*}

which gives the expected remaining lifetime of a subject, given that the subject has survived up to time \(x\). The quantile analogue of the mean residual life function is called the \emph{mean residual quantile function (MRQF)}, and is defined by \citep{NairSankaranBalakrishnan2013}
\begin{equation}
\label{eq:mrqf_def}
M(u)=\frac{1}{1-u}\int_u^1\bigl(Q(p)-Q(u)\bigr)\,dp,
\qquad 0<u<1.
\end{equation}

By \eqref{eq:qepf_def} and \eqref{eq:mrqf_def},
\begin{equation}
\label{eq:qepf_mrq}
M(u)=(P(u)-1)\,Q(u),\qquad 0<u<1.
\end{equation}

\begin{definition}[]
\noindent
One of the basic concepts in reliability theory is the hazard rate, which represents the instantaneous failure of a subject in an interval, given that the subject has survived up to a specified time. The quantile version of the hazard rate is called the \emph{hazard--quantile function}, and is given by
\begin{equation}
\label{eq:hqf_def}
H(u)=\frac{1}{(1-u)q(u)}, \qquad 0<u<1.
\end{equation}
where $q(u)$ denotes the quantile density function, defined as $q(u) = \frac{d}{du} Q(u)$. For more properties of \(H(u)\), one may refer to \citet{NairSankaran2009} and \citet{NairSankaranBalakrishnan2013}.
\end{definition}

Differentiating $(1-u)P(u)Q(u)=\int_u^1 Q(p)\,dp$ gives
\begin{equation}
\label{eq:logderiv}
\frac{q(u)}{Q(u)}
=
\frac{P(u)-1-(1-u)P'(u)}{(1-u)P(u)},
\qquad 0<u<1.
\end{equation}
By \eqref{eq:hqf_def}, \eqref{eq:qepf_mrq} and \eqref{eq:logderiv}, the three measures are linked by
\begin{equation}
\label{eq:hmp_identity}
\frac{H(u)M(u)}{P(u)}
=
\frac{P(u)-1}{P(u)-1-(1-u)P'(u)},
\qquad 0<u<1.
\end{equation}

\begin{proposition}
\label{prop:qepf_recovers_Q}
Let $X$ have finite mean $\mu$. The persistence function $P(u)$ determines the quantile function up to the scale factor $\mu$. Specifically,
\begin{equation}
\label{eq:Q_recovery_from_P}
Q(u)
=
\frac{\mu}{(1-u)P(u)}
\exp\!\left(
-\int_{0}^{u}\frac{dt}{(1-t)P(t)}
\right),
\qquad 0<u<1,
\end{equation}
where the integral is finite for every $u<1$. Since $P$ is invariant under $X\mapsto cX$, no functional of $P$ alone can identify the scale; \eqref{eq:Q_recovery_from_P} therefore determines $Q$ as far as $P$ can.
\end{proposition}

\begin{proof}
Let
\begin{equation*}
W(u)=\int_u^1 Q(p)\,dp,\qquad 0<u<1.
\end{equation*}
By \eqref{eq:qepf_def}, $W(u)=(1-u)P(u)Q(u)$, and $W(0^+)=\mu\in(0,\infty)$. Since $W'(u)=-Q(u)$,
\begin{equation*}
\frac{W'(u)}{W(u)}=-\frac{1}{(1-u)P(u)},\qquad 0<u<1.
\end{equation*}
Integrating from $0$ to $u$ and using $W(0^+)=\mu$,
\begin{equation*}
W(u)=\mu\exp\left(-\int_0^u\frac{dt}{(1-t)P(t)}\right),
\end{equation*}
where the integral is finite because $W(u)>0$ for $u<1$. Dividing by $(1-u)P(u)$ gives \eqref{eq:Q_recovery_from_P}. This completes the proof.
\end{proof}

\begin{theorem}
\label{thm:admissible_P}
Let $P:(0,1)\to(0,\infty)$ be absolutely continuous. Then $P$ is the QEPF of a continuous non-negative random variable with finite mean if and only if
\begin{enumerate}[label=(\roman*),leftmargin=*]
\item $D(u):=P(u)-1-(1-u)P'(u)>0$ for almost every $u\in(0,1)$;
\item $\int_0^u \{(1-t)P(t)\}^{-1}\,dt<\infty$ for every $u<1$, and $\int_0^u \{(1-t)P(t)\}^{-1}\,dt\to\infty$ as $u\to1$.
\end{enumerate}
In that case the quantile function is given by \eqref{eq:Q_recovery_from_P} for an arbitrary $\mu>0$.
\end{theorem}

\begin{proof}
\emph{Necessity.} Since $Q$ is strictly increasing, $q>0$, and \eqref{eq:logderiv} gives
\begin{equation*}
D(u)=(1-u)P(u)\,\frac{q(u)}{Q(u)}>0,
\end{equation*}
which is (i). For (ii), \Cref{prop:qepf_recovers_Q} gives
\begin{equation}
\label{eq:W_from_P}
W(u)=\mu\exp\left(-\int_0^u\frac{dt}{(1-t)P(t)}\right),\qquad 0<u<1.
\end{equation}
Since $W(u)>0$ for $u<1$, the integral is finite. Since $W(1^-)=0$, the integral diverges as $u\to1$.

\emph{Sufficiency.} Given (i) and (ii), define $W$ by \eqref{eq:W_from_P} and set
\begin{equation*}
Q(u)=\frac{W(u)}{(1-u)P(u)}>0,\qquad 0<u<1.
\end{equation*}
Then
\begin{equation*}
W'(u)=-Q(u),\qquad W(0^+)=\mu,\qquad W(1^-)=0,
\end{equation*}
so $W(u)=\int_u^1Q(p)\,dp$ and $\int_0^1Q(p)\,dp=\mu$. Hence the QEPF of $Q$ is
\begin{equation*}
\frac{1}{(1-u)Q(u)}\int_u^1Q(p)\,dp=\frac{W(u)}{(1-u)Q(u)}=P(u).
\end{equation*}
Finally, differentiating $\log Q=\log W-\log(1-u)-\log P$ gives
\begin{equation*}
\frac{q(u)}{Q(u)}=-\frac{1}{(1-u)P(u)}+\frac{1}{1-u}-\frac{P'(u)}{P(u)}=\frac{D(u)}{(1-u)P(u)}>0
\end{equation*}
by (i), so $Q$ is a strictly increasing quantile function. This completes the proof.
\end{proof}

\begin{remark}
\label{rem:P_limit}
The limit $P(1^-)$ describes the upper tail. If $X$ has a finite upper endpoint $\omega$, then $V(u)\to\omega$ and $Q(u)\to\omega$, so $P(1^-)=1$. For an exponential-type tail with $M(u)\to\sigma<\infty$, $P(u)=1+M(u)/Q(u)\to1$ as well. If instead $1-F$ is regularly varying with index $-\rho$, $\rho>1$, then Karamata's theorem \citep{BinghamGoldieTeugels1987} gives $V(u)/Q(u)\to\rho/(\rho-1)>1$. Thus $P(1^-)>1$ signals a Pareto-type upper tail.
\end{remark}

\subsection[Rational specification]{Rational specification of the QEPF}
\label{sec:mobius}

From \eqref{eq:qepf_mrq}, we have

\begin{equation*}
P(u)=1+\frac{M(u)}{Q(u)},\qquad 0<u<1.
\end{equation*}

Since \(P(u)-1=\{a-c+(b-d)u\}/(c+du)\) whenever \(P(u)=(a+bu)/(c+du)\), a rational form of linear-fractional (M\"obius) type for the relative mean residual quantile function \(M(u)/Q(u)\) is equivalent to a linear-fractional form for \(P(u)\). We call this the rational specification of the QEPF. We therefore consider
\begin{equation}
\label{eq:mobius_P_general}
P(u)=\frac{a+bu}{c+du},\qquad 0<u<1,
\end{equation}
where \(a,b,c,d\in\mathbb{R}\) satisfy \(a+bu>0\) and \(c+du>0\) for \(0<u<1\), and \(bc-ad\neq0\). The last condition excludes the degenerate case in which \(P\) is constant. A constant QEPF \(P\equiv k>1\) gives, by \Cref{prop:qepf_recovers_Q}, the Pareto quantile function \(Q(u)=(\mu/k)(1-u)^{-(1-1/k)}\), which has \(Q(0^+)>0\).

\begin{theorem}
\label{thm:general_rational_mrq}
Let \(Q:(0,1)\to(0,\infty)\) be a differentiable, strictly increasing quantile function with finite mean, and assume that its QEPF satisfies
\eqref{eq:mobius_P_general}, where \(a+bu>0\) and \(c+du>0\) on \((0,1)\), \(bc-ad\neq0\), \(a+b>0\) and \(b\neq 0\). Then \Q\, must be of the form
\begin{equation}
\label{eq:mobius_Q_general}
Q(u)
=
K(c+du)(1-u)^{-\xi}(a+bu)^\eta,
\qquad 0<u<1, \text{ for some constant \(K>0\),}
\end{equation}
\begin{equation}
\label{eq:mobius_xi_eta}
\text{where} \qquad
\xi=\frac{(a+b)-(c+d)}{a+b},
\qquad
\eta=-\frac{ab-ad+b^2+bc}{b(a+b)}.
\end{equation}
Conversely, \eqref{eq:mobius_Q_general} is a strictly increasing quantile function with finite mean and QEPF \eqref{eq:mobius_P_general} if and only if \(\xi<1\) and
\begin{equation}
\label{eq:N_positive}
N(u):=\{a-c+(b-d)u\}(c+du)-(bc-ad)(1-u)>0,\qquad 0<u<1.
\end{equation}
If \(c+d\neq 0\), then

\begin{equation*}
Q(u) = K(c+d)(a+b)^\eta(1-u)^{-\xi},
\qquad u\to 1.
\end{equation*}

\end{theorem}

\begin{proof}
From \eqref{eq:mobius_P_general},
\begin{equation*}
P(u)-1=\frac{a-c+(b-d)u}{c+du},\qquad P'(u)=\frac{bc-ad}{(c+du)^2}.
\end{equation*}
Substituting these into \eqref{eq:logderiv} gives
\begin{equation}
\label{eq:qQ_N}
\frac{q(u)}{Q(u)}=\frac{N(u)}{(1-u)(a+bu)(c+du)},
\end{equation}
with $N$ as in \eqref{eq:N_positive}. Decomposing \eqref{eq:qQ_N} into partial fractions, we obtain
\begin{equation*}
\frac{q(u)}{Q(u)}
=
\frac{\xi}{1-u}
+\frac{\eta b}{a+bu}
+\frac{d}{c+du},
\end{equation*}
where \(\xi\) and \(\eta\) are given by \eqref{eq:mobius_xi_eta}. Integrating both sides yields
\begin{equation}
\label{eq:log_Q}
\log Q(u)
=
-\xi\log(1-u)+\eta\log(a+bu)+\log(c+du)+\log K, \qquad K>0
\end{equation}
Exponentiating \Cref{eq:log_Q} gives \eqref{eq:mobius_Q_general}. Since the denominator in \eqref{eq:qQ_N} is positive on $(0,1)$, $Q$ is strictly increasing if and only if \eqref{eq:N_positive} holds.

For the converse, suppose that $Q$ is given by \eqref{eq:mobius_Q_general} with $\xi<1$ and \eqref{eq:N_positive}, and define
\begin{equation*}
G(u):=(1-u)\,\frac{a+bu}{c+du}\,Q(u)=K(a+bu)^{\eta+1}(1-u)^{1-\xi},\qquad 0<u<1.
\end{equation*}
A direct differentiation gives
\begin{equation*}
G'(u)=-Q(u),\qquad 0<u<1.
\end{equation*}
Since $\xi<1$, we have $1-\xi>0$, and because $a+bu>0$ on $(0,1)$ with $a+b>0$,
\begin{equation*}
G(u)\to 0\qquad(u\to 1).
\end{equation*}
Hence
\begin{equation}
\label{eq:G_u}
G(u)=\int_u^1 Q(p)\,dp.
\end{equation}
Dividing \eqref{eq:G_u} by $(1-u)Q(u)$ gives \eqref{eq:mobius_P_general}. The stated asymptotic relation follows from \eqref{eq:mobius_Q_general} by letting $u\to1$ in the factors $(c+du)(a+bu)^{\eta}$. This completes the proof.
\end{proof}

The tail mean, the mean residual quantile function and the hazard--quantile function of the general rational class are also explicit.

\begin{proposition}
\label{prop:general_class_measures}
Under the assumptions of \Cref{thm:general_rational_mrq},

\begin{align*}
V(u) 
&= P(u)Q(u) 
= K(a+bu)^{\eta+1}(1-u)^{-\xi},
\qquad 0<u<1, \\
M(u) 
&= (P(u)-1)Q(u)
= K\{a-c+(b-d)u\}(1-u)^{-\xi}(a+bu)^\eta,
\qquad 0<u<1,\\
H(u) 
&= \frac{(1-u)^\xi(a+bu)^{1-\eta}}{K\,N(u)},
\qquad 0<u<1, \\
\end{align*}
where $N(u)$ is given in \eqref{eq:N_positive}.
\end{proposition}

\begin{proof}
The expression for $V(u)$ follows immediately from $V(u)=P(u)Q(u)$.

Using \eqref{eq:qepf_mrq},

\begin{equation*}
M(u)=(P(u)-1)Q(u).
\end{equation*}

Now substituting 

\begin{equation*}
P(u)-1=\frac{a-c+(b-d)u}{c+du}
\end{equation*}

gives the stated formula for $M(u)$.

We note that 
\begin{equation*}
\frac{q(u)}{Q(u)}
=
\frac{N(u)}{(1-u)(a+bu)(c+du)}.
\end{equation*}

Therefore 
\begin{equation*}
(1-u)q(u)
=
Q(u)\frac{N(u)}{(a+bu)(c+du)}
=
K(1-u)^{-\xi}(a+bu)^{\eta-1}N(u),
\end{equation*}

\begin{equation*}
\text{which yields} \qquad
H(u)=\frac{1}{(1-u)q(u)}
=
\frac{(1-u)^\xi(a+bu)^{1-\eta}}
{K\,N(u)}.
\end{equation*} This completes the proof.
\end{proof}

\begin{proposition}
\label{thm:mobius_to_canonical}
Assume that $Q$ is differentiable and strictly increasing on $(0,1)$ with finite mean $\mu$ and
$Q(0^+):=\lim_{u\downarrow0}Q(u)=0$, and that its QEPF has the rational form \eqref{eq:mobius_P_general}. Then $c=0$, $a>0$, $d>0$, and with
\begin{equation}
\label{eq:gamma_lambda_def}
\gamma:=\frac{b}{a},\qquad \lambda:=\frac{a+b}{d},
\end{equation}
the QEPF takes the form
\begin{equation}
\label{eq:PQEPF}
P(u)=\lambda\,\frac{1+\gamma u}{(1+\gamma)u},\qquad 0<u<1,\qquad \gamma>-1,\ \lambda\ge1.
\end{equation}
Equivalently,
\begin{equation}
\label{eq:odds_linear}
\frac{M(u)}{Q(u)}=(\lambda-1)+\frac{\lambda}{1+\gamma}\cdot\frac{1-u}{u},\qquad 0<u<1.
\end{equation}
The upper endpoint $Q(1^-)$ is finite if and only if $\lambda=1$.
\end{proposition}

\begin{proof}
\emph{Step 1: $c=0$.} Since $Q(0^+)=0$ and $Q$ is strictly increasing with finite mean,
\begin{equation*}
\int_u^1 Q(p)\,dp \to \mu \in (0,\infty)
\qquad \text{and} \qquad
Q(u)\to 0
\qquad (u\downarrow 0),
\end{equation*}
so
\begin{equation*}
P(u)=\frac{1}{(1-u)Q(u)}\int_u^1 Q(p)\,dp\to\infty\qquad(u\downarrow0).
\end{equation*}
Because $c+du>0$ on $(0,1)$, we have $c\ge0$. If $c>0$, then $P$ is bounded near $0$. Hence $c=0$ and $a\neq0$.

\emph{Step 2: reparameterization.} With $c=0$, positivity of $du$ gives $d>0$, and positivity of $a+bu$ near $0$ gives $a>0$. Writing $b=\gamma a$ and $a+b=\lambda d$ as in \eqref{eq:gamma_lambda_def}, we have $a/d=\lambda/(1+\gamma)$, so
\begin{equation*}
P(u)=\frac{a+bu}{du}=\frac{a}{d}\cdot\frac{1+\gamma u}{u}=\lambda\,\frac{1+\gamma u}{(1+\gamma)u},
\end{equation*}
which is \eqref{eq:PQEPF}. Positivity of $a+b$ gives $\gamma>-1$. Identity \eqref{eq:odds_linear} follows from $M(u)/Q(u)=P(u)-1$ (\eqref{eq:qepf_mrq}).

\emph{Step 3: admissibility.} With $c=0$, the function in \eqref{eq:N_positive} is
\begin{equation*}
N(u)=d\{a+(b-d)u^2\}.
\end{equation*}
If $b\ge d$, then $N>0$ on $(0,1)$ and $\lambda=(a+b)/d>1$. If $b<d$, then $N$ is decreasing in $u$, so $N>0$ on $(0,1)$ if and only if
\begin{equation*}
N(1)=d(a+b-d)\ge0,\qquad\text{that is,}\qquad\lambda\ge1.
\end{equation*}
By \Cref{thm:general_rational_mrq}, $Q$ is a strictly increasing quantile function if and only if $\lambda\ge1$.

\emph{Step 4: upper endpoint.} From \eqref{eq:mobius_xi_eta} with $c=0$,
\begin{equation*}
\xi=\frac{a+b-d}{a+b}=1-\frac{1}{\lambda}\in[0,1),
\end{equation*}
so the mean is always finite. Since $Q(u)\asymp(1-u)^{-\xi}$ as $u\to1$, $Q$ is bounded if and only if $\xi=0$, that is, $\lambda=1$. This completes the proof.
\end{proof}

\begin{corollary}[Canonical form]
\label{cor:canonical}
Under the assumptions of \Cref{thm:mobius_to_canonical}, suppose in addition that the upper endpoint $Q(1^-)$ is finite. Then $\lambda=1$, and
\begin{equation}
\label{eq:canonical_P}
P(u)=\frac{1+\gamma u}{(1+\gamma)u},
\qquad
\frac{M(u)}{Q(u)}=\frac{1}{1+\gamma}\cdot\frac{1-u}{u},
\qquad 0<u<1,\ \gamma>-1.
\end{equation}
\end{corollary}

\begin{proof}
By the last statement of \Cref{thm:mobius_to_canonical}, a finite endpoint forces $\lambda=1$. Substituting $\lambda=1$ in \eqref{eq:PQEPF} and \eqref{eq:odds_linear} gives \eqref{eq:canonical_P}. This completes the proof.
\end{proof}

\noindent
The conditions $Q(0^+)=0$ and $Q(1^-)<\infty$ are the boundary conditions for a lifetime that starts at zero and cannot exceed a finite maximal value, as with survival measured within a bounded clinical horizon. Under them, the four coefficients of the rational specification reduce to the single shape parameter $\gamma$, and \eqref{eq:canonical_P} is the canonical form. Its second identity states that the relative mean residual quantile function is proportional to the survival odds. The constant of proportionality $1/(1+\gamma)$ is the only shape information in the model.

\section{The Rational QEPF family}
\label{sec:closed_form}

The ratio $M(u)/Q(u)=P(u)-1$ is the expected residual lifetime beyond the $u$-th quantile, measured in units of that quantile. The Rational QEPF family is defined by one condition on this ratio.

\begin{theorem}[Odds characterization]
\label{thm:odds_law}
Let $X$ be a non-negative random variable with finite mean $\mu$ and a quantile function $Q$ that is continuous and strictly increasing on $(0,1)$. Let $c>0$ and $\gamma=1/c-1>-1$. Then
\begin{equation}
\label{eq:odds_law}
\frac{M(u)}{Q(u)}=c\,\frac{1-u}{u},\qquad 0<u<1,
\end{equation}
if and only if
\begin{equation}
\label{eq:odds_Q}
Q(u)=
\begin{cases}
\mu(1+\gamma)\,u\,(1+\gamma u)^{1/\gamma-1}, & \gamma\neq0,\\[4pt]
\mu\,u\,e^{u}, & \gamma=0,
\end{cases}
\qquad 0<u<1.
\end{equation}
In that case $Q(0^+)=0$ and $Q(1^-)<\infty$.
\end{theorem}

\begin{proof}
Let $W(u)=\int_u^1Q(p)\,dp$. Then $W'=-Q$, $W>0$ on $(0,1)$, $W(0^+)=\mu$ and $W(1^-)=0$. By \eqref{eq:qepf_mrq},
\begin{equation}
\label{eq:odds_W}
\frac{M(u)}{Q(u)}=\frac{W(u)}{(1-u)Q(u)}-1.
\end{equation}
Since $c(1+\gamma)=1$, we have $c+(1-c)u=c(1+\gamma u)$, and
\begin{equation}
\label{eq:odds_pf}
\frac{u}{c(1-u)(1+\gamma u)}=\frac{1}{1-u}-\frac{1}{1+\gamma u}.
\end{equation}

\emph{If.} Take $\gamma\neq0$ and let $Q$ be given by \eqref{eq:odds_Q}. Put $\tilde W(u)=\mu(1-u)(1+\gamma u)^{1/\gamma}$. Then
\begin{equation*}
\tilde W'(u)=-\mu(1+\gamma u)^{1/\gamma-1}\{(1+\gamma u)-(1-u)\}=-Q(u),
\end{equation*}
and $\tilde W(1)=0$, so $W=\tilde W$. Substituting in \eqref{eq:odds_W} gives
\begin{equation*}
\frac{M(u)}{Q(u)}=\frac{1+\gamma u}{(1+\gamma)u}-1=\frac{1-u}{(1+\gamma)u}=c\,\frac{1-u}{u}.
\end{equation*}

\emph{Only if.} Suppose that \eqref{eq:odds_law} holds. By \eqref{eq:odds_W},
\begin{equation*}
\frac{W(u)}{(1-u)Q(u)}=1+c\,\frac{1-u}{u}=\frac{c(1+\gamma u)}{u},
\end{equation*}
so $Q(u)=uW(u)/\{c(1-u)(1+\gamma u)\}$. With $W'=-Q$ and \eqref{eq:odds_pf}, this gives
\begin{equation*}
\frac{d}{du}\log W(u)=-\frac{1}{1-u}+\frac{1}{1+\gamma u},\qquad 0<u<1.
\end{equation*}
The right-hand side is bounded near $u=0$. Integrating from $0$ and using $W(0^+)=\mu$,
\begin{equation*}
W(u)=\mu(1-u)(1+\gamma u)^{1/\gamma},
\end{equation*}
and $Q=-W'$ is \eqref{eq:odds_Q}. The case $\gamma=0$ is the same, with $(1+\gamma u)^{1/\gamma}$ replaced by $e^{u}$.

Finally, \eqref{eq:odds_Q} gives $Q(0^+)=0$ and $Q(1^-)=\mu(1+\gamma)^{1/\gamma}$ ($\mu e$ for $\gamma=0$), which is finite. This completes the proof.
\end{proof}

\begin{theorem}
\label{thm:hmp_char}
Let $X$ satisfy the conditions of \Cref{thm:odds_law}, with $Q$ differentiable. Then $X$ belongs to the Rational QEPF family, that is, $Q$ is given by \eqref{eq:odds_Q} for some $\mu>0$ and $\gamma>-1$, if and only if
\begin{equation}
\label{eq:hmp_char}
\frac{H(u)M(u)}{P(u)}=\frac{u}{1+u},\qquad 0<u<1.
\end{equation}
\end{theorem}

\begin{proof}
Suppose that \eqref{eq:hmp_char} holds. By \eqref{eq:hmp_identity}, $(1+u)\{P(u)-1\}=u\{P(u)-1-(1-u)P'(u)\}$, that is,
\begin{equation*}
\frac{P'(u)}{P(u)-1}=-\frac{1}{u(1-u)}=-\frac{1}{u}-\frac{1}{1-u},\qquad 0<u<1,
\end{equation*}
where $P(u)>1$. Integrating gives $P(u)-1=c(1-u)/u$ for some $c>0$. By \eqref{eq:qepf_mrq}, this is \eqref{eq:odds_law}, and \Cref{thm:odds_law} gives \eqref{eq:odds_Q} with $\gamma=1/c-1$. Conversely, if $Q$ is given by \eqref{eq:odds_Q}, then $P(u)-1=c(1-u)/u$ with $c=1/(1+\gamma)$, so that $P'(u)=-c/u^2$ and $P(u)-1-(1-u)P'(u)=c(1-u)(1+u)/u^2$. Substituting in \eqref{eq:hmp_identity} gives \eqref{eq:hmp_char}. This completes the proof.
\end{proof}

\begin{remark}
Relation \eqref{eq:hmp_char} contains neither $\mu$ nor $\gamma$: the function $H(u)M(u)/P(u)$ is the same for every member of the Rational QEPF family. It fails for the extended family with $\lambda>1$, where $P(u)-1$ contains the constant term $\lambda-1$ (\Cref{thm:mobius_to_canonical}), so \eqref{eq:hmp_char} also separates bounded support from Pareto-type tails.
\end{remark}

\noindent
\Cref{thm:odds_law} assumes neither $Q(0^+)=0$ nor a finite upper endpoint; both follow from \eqref{eq:odds_law}. It places the family in a sequence of models defined by simple forms of the mean residual quantile function:
\begin{enumerate}[label=(\alph*),leftmargin=*]
\item $M(u)$ constant characterizes the exponential distribution \citep{NairSankaranBalakrishnan2013};
\item $M(u)/Q(u)$ constant, that is, a constant QEPF, characterizes the Pareto distribution, which has $Q(0^+)>0$ (\Cref{sec:mobius});
\item $M(u)/Q(u)$ proportional to the survival odds $(1-u)/u$ characterizes the Rational QEPF family.
\end{enumerate}
The uniform distribution on $(0,2\mu)$ is the member with $c=1/2$ ($\gamma=1$). In terms of the QEPF, \eqref{eq:odds_law} reads $P(u)=(1+\gamma u)/\{(1+\gamma)u\}$, which is the canonical form \eqref{eq:canonical_P}. By \Cref{cor:canonical}, it is the only rational specification of the QEPF with $Q(0^+)=0$ and a finite upper endpoint. The rest of this section derives the properties of the family. In the parameterization of \Cref{sec:mobius}, with $\gamma=b/a$, the QEPF is
\begin{equation*}
P^{*}(u)=\frac{a+bu}{(a+b)u},\qquad 0<u<1,
\end{equation*}
with $a>0$ and $a+b>0$. Throughout, an asterisk denotes a function of the Rational QEPF family. The extension to unbounded support is given in \Cref{sec:extension}.

\begin{theorem}
\label{thm:closed_form_quantile}
Assume that
\begin{equation*}
P(u)=\frac{a+bu}{(a+b)u},\qquad 0<u<1,
\end{equation*}
with $a>0$ and $a+b>0$. Then the corresponding quantile function is, 
\begin{equation}
\label{eq:canonical_quantile_ab}
Q^{*}(u)=
\begin{cases}
K_1\,u\,(a+bu)^{a/b-1}, & b\neq 0,\\[4pt]
K_1\,u\,e^{u}, & b=0,
\end{cases}
\qquad 0<u<1.
\end{equation}
where $K_1>0$ is a constant.
$Q^{*}(u)$ is strictly increasing on $(0,1)$, satisfies $Q^{*}(0^+)=0$, and has finite upper endpoint
\begin{equation*}
Q^{*}(1^-)=
\begin{cases}
K_1\,(a+b)^{a/b-1}, & b\neq 0,\\[4pt]
K_1\,e, & b=0.
\end{cases}
\end{equation*}
\end{theorem}

\begin{proof}
Suppose that $b\neq 0$. From \eqref{eq:logderiv},
\begin{equation*}
\frac{q^{*}(u)}{Q^{*}(u)}=
\frac{P^{*}(u)-1-(1-u){P^{*}}'(u)}{(1-u)P^{*}(u)}.
\end{equation*}
With
\begin{equation*}
P(u)=\frac{a+bu}{(a+b)u},
\end{equation*}
a direct calculation gives
\begin{equation*}
\frac{q^{*}(u)}{Q^{*}(u)}=\frac{a(u+1)}{u(a+bu)}
=\frac{1}{u}+\frac{a-b}{a+bu}.
\end{equation*}
Integrating,
\begin{equation*}
\log Q^{*}(u)=\log u+\left(\frac{a}{b}-1\right)\log(a+bu)+\log K,
\end{equation*}
for some $K>0$, and hence
\begin{equation*}
Q^{*}(u)=K_1\,u\,(a+bu)^{a/b-1}.
\end{equation*}

If $b=0$, then $P(u)=1/u$, and \eqref{eq:logderiv} gives
\begin{equation*}
\frac{q^{*}(u)}{Q^{*}(u)}=1+\frac{1}{u}.
\end{equation*}
Thus
\begin{equation*}
\log Q^{*}(u)=u+\log u+\log K_1,
\end{equation*}
so
\begin{equation*}
Q^{*}(u)=K_1\,u\,e^u.
\end{equation*}
In both cases, $Q^{*}(u)\to 0$ as $u\to 0$, so $Q^{*}(0^+)=0$. Also,
\begin{equation*}
\frac{q^{*}(u)}{Q^{*}(u)}>0,\qquad 0<u<1,
\end{equation*}
since $a>0$ and $a+bu>0$ on $(0,1)$. Hence $Q^{*}$ is strictly increasing on $(0,1)$. The endpoint formula follows by letting $u\to 1$. This completes the proof.
\end{proof}

\subsection{Mean normalization and one-shape reparameterization}
\label{sec:mean_reparam}

The multiplicative constant $K>0$ acts only as a scale parameter, so it is
convenient to fix it through the mean condition
\begin{equation*}
\int_0^1 Q^{*}(u)\,du=\mu>0.
\end{equation*}
We then introduce the shape parameter
\begin{equation}
\label{eq:lbl_gamma}
\gamma:=\frac{b}{a},\qquad \gamma>-1.
\end{equation}
The restriction $\gamma>-1$ is equivalent to the conditions $a>0$ and
$a+b>0$.

\noindent
From this point onward, we denote the mean-normalized quantile function by $Q^{*}(u)$.

\begin{proposition}
\label{prop:mean_normalized_family}
For the Rational QEPF family and the reparameterization \eqref{eq:lbl_gamma}, the quantile function can be written as

\begin{equation}
\label{eq:canonical_quantile_gamma}
Q^{*}(u)=
\begin{cases}
\mu(1+\gamma)\,u\,(1+\gamma u)^{1/\gamma-1}, & \gamma\neq 0,\\[4pt]
\mu\,u\,e^u, & \gamma=0,
\end{cases}
\qquad 0<u<1.
\end{equation}
Its upper endpoint is
\begin{equation}
\label{eq:canonical_endpoint_gamma}
\omega:=Q^{*}(1^-)=
\begin{cases}
\mu(1+\gamma)^{1/\gamma}, & \gamma\neq 0,\\[4pt]
\mu e, & \gamma=0,
\end{cases}
\end{equation}
and
\begin{equation}
\label{eq:q_over_Q_gamma}
\frac{q^{*}(u)}{Q^{*}(u)}=\frac{1+u}{u(1+\gamma u)},
\qquad 0<u<1.
\end{equation}
In particular, \(q^{*}(u)>0\) on \((0,1)\) for all \(\gamma>-1\).
\end{proposition}

\begin{proof}
Suppose first that $\gamma\neq 0$, so that $b=\gamma a$ with $a>0$. From
\Cref{thm:closed_form_quantile},
\begin{equation*}
Q^{*}(u)=K\,u\,(a+bu)^{a/b-1}
=K_2\,a^{1/\gamma-1}u(1+\gamma u)^{1/\gamma-1}.
\end{equation*}
Writing $C=K_2\,a^{1/\gamma-1}$ gives
\begin{equation*}
Q^{*}(u)=C\,u(1+\gamma u)^{1/\gamma-1}.
\end{equation*}
Since
\begin{equation*}
\int_0^1 u(1+\gamma u)^{1/\gamma-1}\,du=\frac{1}{1+\gamma},
\qquad \gamma>-1,\ \gamma\neq 0,
\end{equation*}
the choice $C=\mu(1+\gamma)$ yields \eqref{eq:canonical_quantile_gamma} for $\gamma\neq 0$.

If $\gamma=0$, then $b=0$, and \Cref{thm:closed_form_quantile} gives 
\begin{equation*}
Q^{*}(u)=K_2ue^u.
\end{equation*}

\noindent
Since

\begin{equation*}
\int_0^1 ue^u\,du=1,
\end{equation*}

taking $K=\mu$ gives the second line of \eqref{eq:canonical_quantile_gamma}.

The endpoint formula \eqref{eq:canonical_endpoint_gamma} follows by letting $u\to 1$. Finally, substituting $b=\gamma a$ into the log-derivative identity gives
\begin{equation*}
\frac{q^{*}(u)}{Q^{*}(u)}=\frac{1+u}{u(1+\gamma u)},
\end{equation*}
which is \eqref{eq:q_over_Q_gamma}. Since $\gamma>-1$, one has $1+\gamma u>0$ on $(0,1)$, so \(q^{*}(u)>0\) throughout. This completes the proof.
\end{proof}

\noindent
Thus the Rational QEPF family has bounded support for every $\gamma>-1$, in agreement with \Cref{cor:canonical}.

\subsection{The induced distribution family: structure and regimes}
\label{subsec:family_summary}

The Rational QEPF family is a two-parameter family indexed by
$(\mu,\gamma)$ with $\mu>0$ and $\gamma>-1$. The model is defined directly
through its quantile function $Q^{*}(u)$ and quantile density
$q^{*}(u)$. 
\begin{equation}
\label{eq:density_from_quantile_density}
f(Q^{*}(u))=\frac{1}{q^{*}(u)}.
\end{equation}

\begin{proposition}
\label{prop:family_regimes}
Let $\gamma>-1$. The Rational QEPF family satisfies the following properties.

\begin{enumerate}[label=(\roman*),leftmargin=*]

\item\textbf{Quantile-domain identities.}
For $0<u<1$,
\begin{equation*}
P^{*}(u)=\frac{1+\gamma u}{(1+\gamma)u},\qquad
V^{*}(u)=P^{*}(u)Q^{*}(u),\qquad
M^{*}(u)=(P^{*}(u)-1)Q^{*}(u).
\end{equation*}

\item\textbf{Support and endpoint.}
The upper endpoint is finite and is given by
\begin{equation*}
\omega:=Q^{*}(1^-)=
\begin{cases}
\mu(1+\gamma)^{1/\gamma}, & \gamma\neq 0,\\
\mu e, & \gamma=0.
\end{cases}
\end{equation*}

\item\textbf{Quantile density and monotonicity.}
The quantile density is
\begin{equation}
\label{eq:q_gamma_explicit}
q^{*}(u)=
\begin{cases}
\mu(1+\gamma)(1+u)(1+\gamma u)^{1/\gamma-2}, & \gamma\neq 0,\\[4pt]
\mu e^u(1+u), & \gamma=0,
\end{cases}
\qquad 0<u<1.
\end{equation}
It is increasing for $-1<\gamma<1$, constant for $\gamma=1$, and decreasing for
$\gamma>1$. Consequently, the density $f$ is monotone.
\end{enumerate}
\end{proposition}

\begin{proof}
The proof for (i) and (ii) are direct. To prove (iii), differentiating \eqref{eq:canonical_quantile_gamma} gives \eqref{eq:q_gamma_explicit}. For $\gamma\neq 0$,
\begin{equation*}
\frac{d}{du}\log q^{*}(u)
=
\frac{1}{1+u}+\frac{1-2\gamma}{1+\gamma u}
=
\frac{(1-\gamma)(u+2)}{(1+u)(1+\gamma u)}.
\end{equation*}
Since $\gamma>-1$, the denominator is positive on $(0,1)$, so the sign is determined by $1-\gamma$. Hence $q^{*}(u)$ is increasing for $-1<\gamma<1$, constant for $\gamma=1$, and decreasing for $\gamma>1$. By \eqref{eq:density_from_quantile_density}, the density $f$, whenever it exists, is therefore monotone. This completes the proof.
\end{proof}

\noindent
By (i), $P^{*}(u)=1/\{(1+\gamma)u\}+\gamma/(1+\gamma)$ is strictly decreasing on $(0,1)$, with $P^{*}(u)\to\infty$ as $u\to0$ and $P^{*}(u)\to1$ as $u\to1$. The tail mean and the mean residual quantile function are
\begin{equation}
\label{eq:V_gamma}
V^{*}(u)=
\begin{cases}
\mu(1+\gamma u)^{1/\gamma}, & \gamma\neq 0,\\[4pt]
\mu e^u, & \gamma=0,
\end{cases}
\qquad 0<u<1,
\end{equation}
\begin{equation}
\label{eq:M_gamma}
M^{*}(u)=
\begin{cases}
\mu(1-u)(1+\gamma u)^{1/\gamma-1}, & \gamma\neq 0,\\[4pt]
\mu(1-u)e^u, & \gamma=0,
\end{cases}
\qquad 0<u<1.
\end{equation}
If $\gamma\ge 0$, then $M^{*}(u)$ is strictly decreasing on $(0,1)$. If $-1<\gamma<0$, then $M^{*}(u)$ is unimodal with unique maximum at $u_M^{*}=-\gamma\in(0,1)$. By \eqref{eq:hqf_def} and \eqref{eq:q_gamma_explicit}, the hazard--quantile function \citep{NairSankaranBalakrishnan2013} of the family is
\begin{equation}
\label{eq:H_gamma}
H^{*}(u)=
\begin{cases}
\dfrac{(1+\gamma u)^{2-1/\gamma}}{\mu(1+\gamma)(1-u^2)}, & \gamma\neq 0,\\[10pt]
\dfrac{1}{\mu e^u(1-u^2)}, & \gamma=0,
\end{cases}
\qquad 0<u<1.
\end{equation}
If $\gamma\ge 1/2$, then $H^{*}(u)$ is strictly increasing on $(0,1)$. If $-1<\gamma<1/2$, then $H^{*}(u)$ is bathtub-shaped with unique minimum at $u_H^{*}=-1+\sqrt{2(1-\gamma)}\in(0,1)$.

\subsection{Relaxing the upper boundary condition}
\label{sec:extension}

The Rational QEPF family has bounded support by construction, whereas many survival outcomes have no upper limit. Dropping the condition $Q(1^-)<\infty$ from \Cref{cor:canonical}, \Cref{thm:mobius_to_canonical} shows that the QEPF becomes the multiple $\lambda\ge1$ of the Rational QEPF given in \eqref{eq:PQEPF}. We call this three-parameter family the \emph{extended Rational QEPF} family. It contains the Rational QEPF family as the case $\lambda=1$, and for $\lambda>1$ its relative mean residual quantile function \eqref{eq:odds_linear} is a constant plus a multiple of the survival odds.

\begin{theorem}
\label{thm:PQEPF_quantile}
Let $\mu>0$, $\gamma>-1$ and $\lambda\ge1$. The quantile function with mean $\mu$ and QEPF \eqref{eq:PQEPF} is
\begin{equation}
\label{eq:PQEPF_Q}
Q(u)=
\begin{cases}
\dfrac{\mu(1+\gamma)}{\lambda}\,u\,(1-u)^{-(1-1/\lambda)}(1+\gamma u)^{1/(\gamma\lambda)-1}, & \gamma\neq0,\\[10pt]
\dfrac{\mu}{\lambda}\,u\,(1-u)^{-(1-1/\lambda)}\,e^{u/\lambda}, & \gamma=0,
\end{cases}
\qquad 0<u<1,
\end{equation}
and its tail integral is
\begin{equation}
\label{eq:PQEPF_W}
W(u)=\int_u^1Q(p)\,dp=
\begin{cases}
\mu(1-u)^{1/\lambda}(1+\gamma u)^{1/(\gamma\lambda)}, & \gamma\neq0,\\[4pt]
\mu(1-u)^{1/\lambda}e^{u/\lambda}, & \gamma=0.
\end{cases}
\end{equation}
\end{theorem}

\begin{proof}
Take $\gamma\neq0$, and let $W$ be given by \eqref{eq:PQEPF_W}. Differentiating $\log W$,
\begin{equation*}
\frac{W'(u)}{W(u)}=-\frac{1}{\lambda(1-u)}+\frac{1}{\lambda(1+\gamma u)}
=-\frac{(1+\gamma)u}{\lambda(1-u)(1+\gamma u)}.
\end{equation*}
Multiplying by $W(u)$,
\begin{equation*}
W'(u)=-\frac{\mu(1+\gamma)}{\lambda}\,u\,(1-u)^{-(1-1/\lambda)}(1+\gamma u)^{1/(\gamma\lambda)-1}=-Q(u),
\end{equation*}
with $Q$ as in \eqref{eq:PQEPF_Q}. Since $1/\lambda>0$,
\begin{equation*}
W(0^+)=\mu,\qquad W(1^-)=0.
\end{equation*}
Hence $W(u)=\int_u^1Q(p)\,dp$, $\int_0^1Q(p)\,dp=\mu$, and the QEPF of $Q$ is
\begin{equation*}
P(u)=\frac{W(u)}{(1-u)Q(u)}=\lambda\,\frac{1+\gamma u}{(1+\gamma)u}.
\end{equation*}
The case $\gamma=0$ is identical, with $(1+\gamma u)^{1/(\gamma\lambda)}$ replaced by its limit $e^{u/\lambda}$. Uniqueness follows from \Cref{prop:qepf_recovers_Q}. This completes the proof.
\end{proof}

\begin{proposition}
\label{prop:PQEPF_properties}
For the extended Rational QEPF family with parameters $(\mu,\gamma,\lambda)$:
\begin{enumerate}[label=(\roman*),leftmargin=*]
\item \textbf{Tail mean and MRQ.} $V(u)=W(u)/(1-u)$ and $M(u)=V(u)-Q(u)$, with $W$ in \eqref{eq:PQEPF_W}.
\item \textbf{Quantile density.} For $\gamma\neq0$,
\begin{equation}
\label{eq:PQEPF_q}
q(u)=\frac{\mu(1+\gamma)}{\lambda^2}(1-u)^{-(2-1/\lambda)}(1+\gamma u)^{1/(\gamma\lambda)-2}\bigl\{\lambda+(\gamma\lambda-1-\gamma)u^2\bigr\},
\end{equation}
and the hazard--quantile function is $H(u)=1/\{(1-u)q(u)\}$. The factor in braces equals $\lambda$ at $u=0$ and $(\lambda-1)(1+\gamma)\ge0$ at $u=1$, so $q>0$ on $(0,1)$.
\item \textbf{Support and tail.} If $\lambda=1$, the support is $(0,\omega)$ with $\omega=\mu(1+\gamma)^{1/\gamma}$ ($\omega=\mu e$ when $\gamma=0$). If $\lambda>1$, the support is $(0,\infty)$ and $Q(u)\sim C(1-u)^{-(1-1/\lambda)}$ as $u\to1$, with $C=\mu\lambda^{-1}(1+\gamma)^{1/(\gamma\lambda)}$ ($C=\mu\lambda^{-1}e^{1/\lambda}$ when $\gamma=0$). Hence the upper tail is of Pareto type with tail index $\lambda/(\lambda-1)$, and $\mathbb{E}(X^r)<\infty$ if and only if $r<\lambda/(\lambda-1)$. In particular, the variance is finite if and only if $\lambda<2$.
\item \textbf{Behaviour at the origin.} $Q(u)\sim \mu\lambda^{-1}(1+\gamma)u$ as $u\downarrow0$, so the density satisfies $f(0^+)=\lambda/\{\mu(1+\gamma)\}$.
\item \textbf{Lorenz curve and TTT transform.} $L(u)=\mu^{-1}\int_0^uQ(p)\,dp=1-W(u)/\mu$ and $T(u)=\int_0^u(1-p)q(p)\,dp=\mu-W(u)+(1-u)Q(u)$, both in closed form \citep{NairSankaranBalakrishnan2013}.
\end{enumerate}
\end{proposition}

\begin{proof}
\begin{enumerate}[label=(\roman*),leftmargin=*]
\item This follows from \Cref{thm:PQEPF_quantile} and the definitions of $V$ and $M$.
\item Differentiating $\log Q$ in \eqref{eq:PQEPF_Q} gives
\begin{equation*}
\frac{q(u)}{Q(u)}=\frac{1}{u}+\frac{1-1/\lambda}{1-u}+\frac{1/\lambda-\gamma}{1+\gamma u}
=\frac{\lambda+(\gamma\lambda-1-\gamma)u^2}{\lambda\,u(1-u)(1+\gamma u)}.
\end{equation*}
Multiplying by $Q(u)$ gives \eqref{eq:PQEPF_q}.
\item The exponent of $(1-u)$ in \eqref{eq:PQEPF_Q} is $-(1-1/\lambda)$, which vanishes exactly when $\lambda=1$. For $\lambda>1$, letting $u\to1$ in the remaining factors gives the constant $C$. For the moments,
\begin{equation*}
\mathbb{E}(X^r)=\int_0^1Q(u)^r\,du<\infty
\iff \int^1(1-u)^{-r(1-1/\lambda)}\,du<\infty
\iff r\Bigl(1-\frac{1}{\lambda}\Bigr)<1,
\end{equation*}
which is $r<\lambda/(\lambda-1)$.
\item From \eqref{eq:PQEPF_q}, $q(0^+)=\mu(1+\gamma)/\lambda$. Since $f(Q(u))=1/q(u)$ and $Q(0^+)=0$,
\begin{equation*}
f(0^+)=\frac{1}{q(0^+)}=\frac{\lambda}{\mu(1+\gamma)}.
\end{equation*}
\item By \Cref{thm:PQEPF_quantile}, $\int_0^uQ(p)\,dp=\mu-W(u)$, which gives $L(u)$. Integrating by parts and using $Q(0^+)=0$,
\begin{equation*}
T(u)=\int_0^u(1-p)q(p)\,dp=(1-u)Q(u)+\int_0^uQ(p)\,dp=\mu-W(u)+(1-u)Q(u).
\end{equation*}
\end{enumerate} This completes the proof.
\end{proof}

\begin{remark}[Interpretation of the parameters]
\label{rem:interpretation}
$\mu$ is the mean. $\lambda=P(1^-)$ is the limit of the QEPF as $u\to1$: by \Cref{rem:P_limit}, $\lambda=1$ corresponds to bounded support, and $\lambda>1$ to a Pareto-type tail in which the mean of the outcomes above a high quantile is about $\lambda$ times that quantile. In extreme value analysis, $P(u)$ is called the tail persistence, and under a regularly varying tail with index $-\rho$, $\rho>1$, its limit is $\rho/(\rho-1)$ (\Cref{rem:P_limit}; \citealp{Rawat2026}); for the extended Rational QEPF, $\rho=\lambda/(\lambda-1)$, so $\lambda$ is this limit. Finally, $\alpha:=\lambda/(1+\gamma)=\mu f(0^+)$ is the early-failure intensity relative to the mean, and it is the coefficient of the odds in \eqref{eq:odds_linear}. Thus $\gamma$ mainly controls the body of the distribution, and $\lambda$ controls the tail. The case $\lambda=1$ recovers \eqref{eq:canonical_quantile_gamma}.
\end{remark}

\section{A log transformation and L-moments}
\label{sec:characterizations}

This section shows that a logarithmic transformation of the upper tail maps the Rational QEPF family into a fractional-linear hazard--quantile class, and then studies its L-moments.

\begin{theorem}
\label{thm:tail_log_transform}
Let $X$ follow the bounded $(\mu,\gamma)$ family ($\lambda=1$) with $\gamma>-1$ and upper endpoint
\begin{equation*}
\omega=Q^{*}(1^-).
\end{equation*}
Define
\begin{equation*}
Z=-\log\!\left(\frac{X}{\omega}\right)\in[0,\infty).
\end{equation*}
Then the quantile function of $Z$ is
\begin{equation}
\label{eq:QZ_general}
Q_Z(u)=-\log\!\left(\frac{Q^{*}(1-u)}{\omega}\right),
\qquad 0<u<1.
\end{equation}
\begin{equation}
\label{eq:QZ_gamma}
Q_Z(u)=
\begin{cases}
-\log(1-u)
+\dfrac{1-\gamma}{\gamma}
\left[
\log(1+\gamma)-\log\!\bigl(1+\gamma(1-u)\bigr)
\right], & \gamma\neq 0, \\[10pt]
-\log(1-u)+u, & \gamma=0,
\end{cases}
\qquad 0<u<1.
\end{equation}
The corresponding hazard--quantile function is
\begin{equation}
\label{eq:HZ_gamma}
H_Z(u)=\frac{1}{(1-u)Q_Z'(u)}
      =\frac{1+\gamma(1-u)}{2-u},
\qquad 0<u<1.
\end{equation}
In the boundary case $\gamma=0$,
\begin{equation}
\label{eq:MZ_gamma0}
M_Z(u)=\frac{3-u}{2}.
\end{equation}
\end{theorem}

\begin{proof}
Since $g(x)=-\log(x/\omega)$ is strictly decreasing on $(0,\omega]$, quantiles
under decreasing transformations satisfy
\begin{equation*}
Q_Z(u)=g(Q^{*}(1-u))
      =-\log\!\left(\frac{Q^{*}(1-u)}{\omega}\right),
\qquad 0<u<1,
\end{equation*}
which gives \eqref{eq:QZ_general}.

For $\gamma\neq 0$, using \eqref{eq:canonical_quantile_gamma} and
\eqref{eq:canonical_endpoint_gamma},
\begin{equation*}
Q^{*}(v)=\mu(1+\gamma)\,v\,(1+\gamma v)^{1/\gamma-1},
\qquad
\omega=\mu(1+\gamma)^{1/\gamma}.
\end{equation*}
Substituting $v=1-u$ gives
\begin{equation*}
\frac{Q^{*}(1-u)}{\omega}
=
(1-u)(1+\gamma)^{1-1/\gamma}
\bigl(1+\gamma(1-u)\bigr)^{1/\gamma-1}.
\end{equation*}
Taking $-\log$ yields the first line of \eqref{eq:QZ_gamma}.

If $\gamma=0$, then
\begin{equation*}
Q^{*}(v)=\mu v e^v,
\qquad
\omega=\mu e,
\end{equation*}
so
\begin{equation*}
\frac{Q^{*}(1-u)}{\omega}=(1-u)e^{-u},
\end{equation*}
which gives the second line of \eqref{eq:QZ_gamma}. Differentiating \eqref{eq:QZ_gamma}, for $\gamma\neq 0$ we obtain
\begin{equation*}
Q_Z'(u)
=
\frac{1}{1-u}
+\frac{1-\gamma}{1+\gamma(1-u)}
=
\frac{2-u}{(1-u)\{1+\gamma(1-u)\}}.
\end{equation*}
Hence
\begin{equation*}
H_Z(u)=\frac{1}{(1-u)Q_Z'(u)}
      =\frac{1+\gamma(1-u)}{2-u},
\end{equation*}
which is \eqref{eq:HZ_gamma}. The case $\gamma=0$ is included directly in the same formula. Finally, when $\gamma=0$, the transformed model has inverse-linear
hazard--quantile function
\begin{equation*}
H_Z(u)=\frac{1}{2-u},
\end{equation*}
for which the corresponding mean residual quantile function is
\begin{equation*}
M_Z(u)=\frac{3-u}{2}.
\end{equation*}
This gives \eqref{eq:MZ_gamma0}. Since \(H_Z(u)=1/(2-u)\), this is the inverse-linear hazard--quantile model, whose mean residual quantile function is \(M_Z(u)=(3-u)/2\); see \citet{MidhuSankaranNair2013}. This completes the proof.
\end{proof}

\begin{remark}
\label{rem:transform_connection}
Writing \eqref{eq:HZ_gamma} as $H_Z(u)=k(1+Au)/(1+Bu)$ with $k=(1+\gamma)/2$, $A=-\gamma/(1+\gamma)$ and $B=-1/2$ shows that $H_Z$ has the bilinear form of \citet{SankaranThomasMidhu2015}. Their parameter range $k>0$, $B\ge-1$, $A>0$ corresponds to $-1<\gamma<0$; for $\gamma>0$ we have $-1<A<0$, for which their quantile density is still positive, so the bilinear form extends to these values. In the boundary case
$\gamma=0$, it reduces to the inverse-linear form
\begin{equation*}
H_Z(u)=\frac{1}{2-u},
\end{equation*}
and \eqref{eq:MZ_gamma0} shows that the corresponding mean residual quantile
function is linear. Thus the transformed family links the proposed model to
earlier linear mean residual quantile constructions
\citep{MidhuSankaranNair2013}.
\end{remark}

\subsection{L-moments and L-moment ratio structure}
\label{sec:lmom}

L-moments provide an alternative set of distributional summaries. L-moments are defined as linear combinations of order statistics and possess an
important structural advantage over conventional power moments in terms of
existence. Specifically, the population L-moments are well defined for any random
variable with a finite mean, and the existence of the first L-moment
\(L_1=\mathbb{E}(X)\) implies the existence of all higher-order L-moments
\(L_r\), \(r\ge 2\) \citep{Hosking1990}. Consequently, the L-moment ratios are well defined whenever the mean exists, even in situations where higher conventional moments may be infinite. This property makes L-moments particularly well suited for characterizing skewness and tail behaviour of distributions with moderate or heavy tails.

\noindent
For the Rational QEPF family defined in \eqref{eq:canonical_quantile_gamma}, the first
L-moment is
\begin{equation}
L_1 = \int_0^1 Q^{*}(u)\,du = \mu.
\label{eq:L1_gamma}
\end{equation}

\noindent
Using the endpoint formula \eqref{eq:canonical_endpoint_gamma}, for $\gamma\neq 0$, the second and third L-moments are

\begin{align}
L_2
&= \int_0^1 (2u-1)\,Q^{*}(u)\,du
 = \frac{2(1+\gamma)\omega-(2\gamma+5)\mu}{2\gamma+1}.
\label{eq:L2_gamma}\\
L_3
&= \int_0^1 (6u^2-6u+1)\,Q^{*}(u)\,du \\
&= \mu\frac{6\gamma^2+41\gamma+49}{(2\gamma+1)(3\gamma+1)}
 - \omega\frac{6(\gamma^2+4\gamma+3)}{(2\gamma+1)(3\gamma+1)}.
\label{eq:L3_gamma}
\end{align}

\noindent
In the boundary case $\gamma=0$, these reduce to
\begin{equation}
L_2=\mu(2e-5),
\qquad
L_3=\mu(49-18e).
\end{equation}

\noindent
The fourth L-moment is
\begin{equation}
\label{eq:L4_gamma}
\begin{aligned}
L_4
&= \int_0^1 (20u^3-30u^2+12u-1)\,Q^{*}(u)\,du \\[4pt]
&=
\begin{cases}
\dfrac{
\bigl(24\gamma^3+228\gamma^2+456\gamma+252\bigr)\,\omega
-
\bigl(24\gamma^3+314\gamma^2+897\gamma+685\bigr)\,\mu
}{
(2\gamma+1)(3\gamma+1)(4\gamma+1)
},
& \gamma\neq 0, \\[10pt]
\mu(252e-685), & \gamma=0.
\end{cases}
\end{aligned}
\end{equation}

\begin{remark}
\label{rem:lmoment_singularities}

The closed-form expressions \eqref{eq:L2_gamma}--\eqref{eq:L4_gamma} have
apparent singularities at $\gamma=-1/2$, $-1/3$, and $-1/4$. These are
removable singularities of the algebraic representation rather than of the
underlying L-moments. At these points, the corresponding values are obtained by
continuity from the integral definitions \eqref{eq:L2_gamma}--\eqref{eq:L4_gamma}.
\end{remark}

\noindent
The L-moment ratios are
\begin{equation}
\tau_2=\frac{L_2}{L_1},\qquad
\tau_3=\frac{L_3}{L_2},\qquad
\tau_4=\frac{L_4}{L_2}.
\label{eq:tau_defs}
\end{equation}

\noindent
For $\gamma\neq 0$, write

\begin{equation*}
A(\gamma)=(1+\gamma)^{1/\gamma}.
\end{equation*}

Then,
\begin{equation}
\tau_2
=
\frac{2(1+\gamma)A(\gamma)-(2\gamma+5)}{2\gamma+1},
\label{eq:tau2_gamma}
\end{equation}

\noindent
The corresponding expressions for $\gamma\neq 0$ are
\begin{align}
\tau_3
&=
\frac{
6\gamma^2+41\gamma+49
-6A(\gamma)(\gamma^2+4\gamma+3)
}{
(3\gamma+1)\bigl[2(1+\gamma)A(\gamma)-(2\gamma+5)\bigr]
},
\label{eq:tau3_gamma}
\\[0.75em]
\tau_4
&=
\frac{
A(\gamma)\bigl(24\gamma^3+228\gamma^2+456\gamma+252\bigr)
-
\bigl(24\gamma^3+314\gamma^2+897\gamma+685\bigr)
}{
(3\gamma+1)(4\gamma+1)
\bigl[2(1+\gamma)A(\gamma)-(2\gamma+5)\bigr]
}.
\label{eq:tau4_gamma}
\end{align}

\noindent
with boundary values
\begin{equation*}
\tau_2(0)=2e-5,\qquad
\tau_3(0)=\frac{49-18e}{2e-5},\qquad
\tau_4(0)=\frac{252e-685}{2e-5}.
\end{equation*}

\noindent
Thus both $(\tau_2,\tau_3)$ and $(\tau_3,\tau_4)$ lie on smooth
one-parameter curves indexed by $\gamma>-1$, with the special values in
Remark~\ref{rem:lmoment_singularities} interpreted by continuity. For $\lambda>1$, \Cref{prop:PQEPF_properties}(iii) gives a Pareto-type tail with index $\lambda/(\lambda-1)>1$, so the mean is finite for every $\lambda\ge1$. By the result of \citet{Hosking1990} quoted above, all L-moments, and hence $\tau_3$ and $\tau_4$, then exist for every member of the extended Rational QEPF family, including the members with $\lambda\ge2$ whose variance is infinite. They have no elementary closed form for $\lambda>1$. Integrating by parts against the shifted Legendre polynomials $P_{r-1}^{*}$ and using $W(0^+)=\mu$, $W(1^-)=0$ gives
\begin{equation*}
L_r=(-1)^{r-1}\mu+\int_0^1W(u)\,P_{r-1}^{*\prime}(u)\,du,\qquad r\ge2,
\end{equation*}
with $W$ in \eqref{eq:PQEPF_W}. The integrand is bounded, so the integral is computed accurately by standard quadrature; at $\lambda=1$ it reproduces \eqref{eq:tau3_gamma}--\eqref{eq:tau4_gamma}. \Cref{fig:tau3tau4} adds the curves for $\lambda\in\{1.2,1.5,2,3\}$. Increasing $\lambda$ moves the curve towards larger L-skewness and L-kurtosis: at $\gamma=0$, $(\tau_3,\tau_4)$ moves from $(0.162,0.016)$ at $\lambda=1$ to $(0.295,0.120)$, $(0.426,0.240)$, $(0.558,0.380)$ and $(0.696,0.548)$. For $\lambda>1$ the curves are not monotone in $\gamma$: $\tau_3$ reaches its minimum at $\gamma$ between about $5$ and $11$ and then increases again as $\gamma\to\infty$, so that, unlike the case $\lambda=1$, $\gamma$ does not order the members by L-skewness. In
particular, when $\gamma=1$,

\begin{equation*}
Q^{*}(u)=2\mu u,
\end{equation*}

so the family contains the uniform distribution on $(0,2\mu)$ as an exact
member, and the corresponding L-moment ratio point is $(\tau_3,\tau_4)=(0,0)$.

\begin{figure}[ht]
\centering
\begin{minipage}{0.48\textwidth}
\centering
\includegraphics[width=\textwidth]{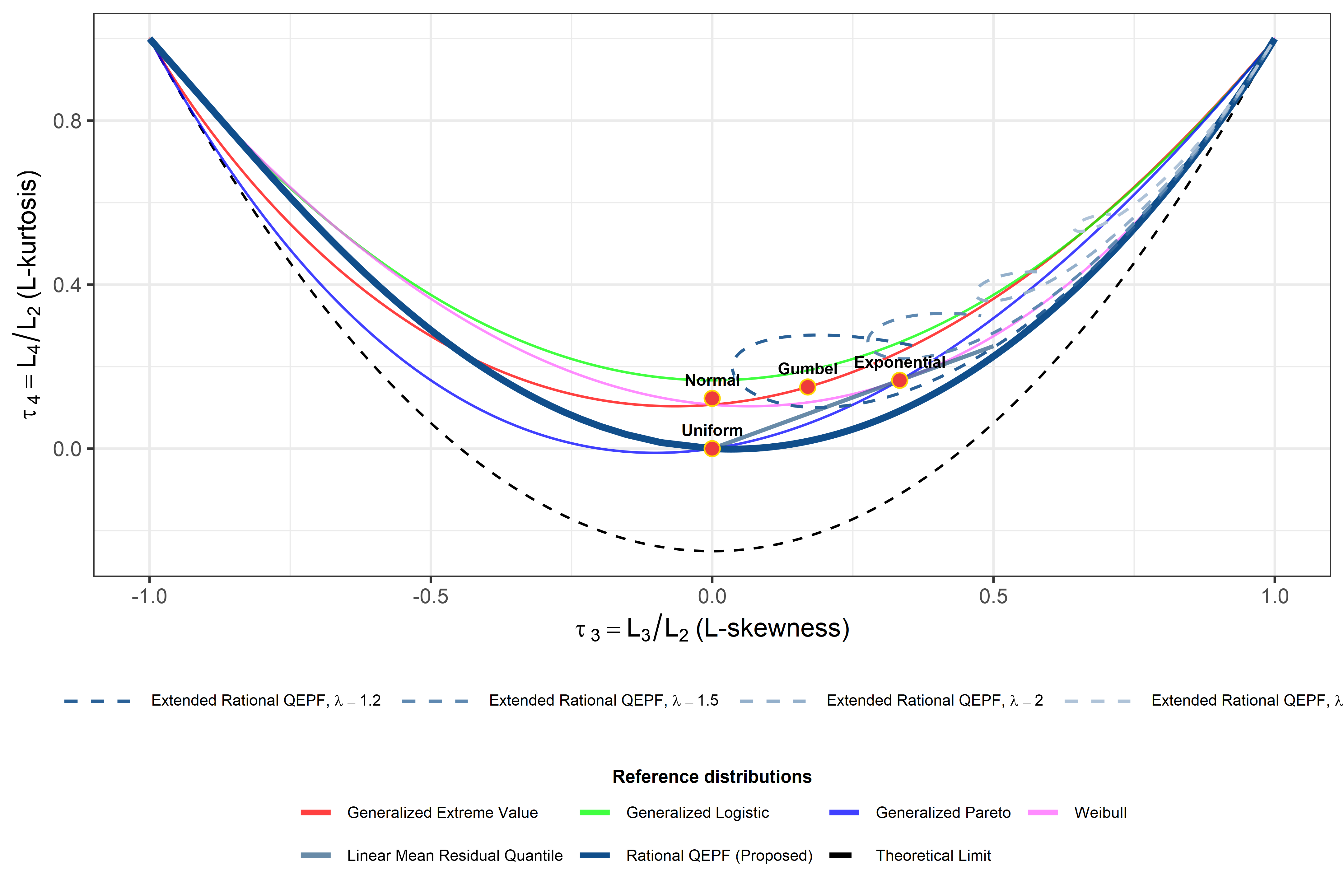}
\caption{L-moment ratio diagram $(\tau_3,\tau_4)$ for the proposed
Rational QEPF family ($\lambda=1$, solid) and the extended Rational QEPF family for $\lambda\in\{1.2,1.5,2,3\}$ (dashed), together with selected reference families\protect\sagenew{. Source: L-moment ratios of Section~\ref{sec:characterizations}; reference loci from \citet{Hosking1990}}.}
\label{fig:tau3tau4}
\end{minipage}
\hfill
\begin{minipage}{0.48\textwidth}
\centering
\includegraphics[width=\textwidth]{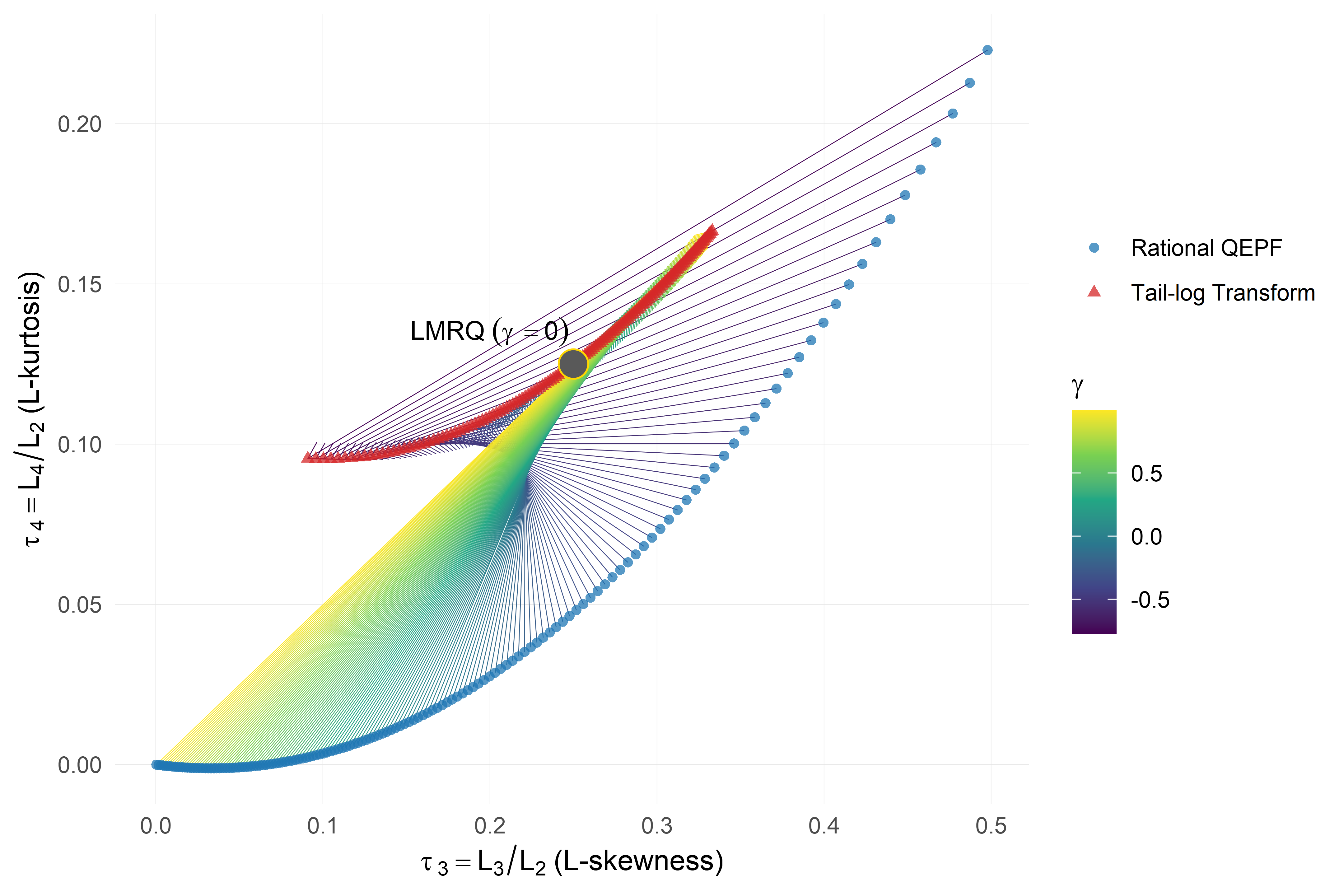}
\caption{L-moment ratio mapping $(\tau_3,\tau_4)$ comparing the
Rational QEPF family and its tail--log transform. \protect\sagenew{Source: L-moment ratios of Section~\ref{sec:characterizations}.}}
\label{fig:tau3tau4Z}
\end{minipage}
\end{figure}

\Cref{fig:tau3tau4} displays the L-moment ratio diagram
$(\tau_3,\tau_4)$ for the proposed Rational QEPF family together with several
classical reference families. The proposed curve passes through the uniform
distribution, confirming the preceding observation. 
The reference families in the diagram follow the L-moment ratio diagram of \citet{Hosking1990}. They comprise the generalized Pareto (GPA),
generalized extreme value (GEV), and generalized logistic (GLO) families,
together with the special points corresponding to the uniform, exponential,
Gumbel, logistic, and normal distributions. We also include the Weibull family and the linear mean residual quantile (LMRQ) family \citep{MidhuSankaranNair2013} for comparison.

\medskip
The Rational QEPF family therefore traces a smooth one-parameter curve in the $(\tau_3,\tau_4)$ plane. \Cref{fig:tau3tau4Z} shows the corresponding L-moment ratio mapping for the tail--log transformed family. Each arrow represents the image of a common shape parameter $\gamma$, mapping the point $(\tau_3,\tau_4)$ of the original Rational QEPF family to the corresponding point of the transformed family. The transformation shifts the family toward smaller values of L-skewness and into a narrower region of the $(\tau_3,\tau_4)$ plane.

\begin{remark}
At $\gamma=0$ the transformed variable $Z$ of \Cref{thm:tail_log_transform} has $H_Z(u)=1/(2-u)$ and $M_Z(u)=(3-u)/2$, so $Z$ is the member of the linear mean residual quantile family of \citet{MidhuSankaranNair2013} with $M(0)=3/2$ and slope $-1/2$. The link to that family therefore holds exactly at $\gamma=0$; for other values of $\gamma$, $Z$ lies in the fractional-linear hazard--quantile class (\Cref{rem:transform_connection}). \Cref{fig:tau3tau4Z} shows the L-moment ratios of the transformed family.
\end{remark}

\section{Inference and estimation}
\label{sec:inference}

We write the likelihood for the extended family, with parameter
\begin{equation*}
\psi=(\mu,\gamma,\lambda), \qquad \mu>0,\ \gamma>-1,\ \lambda\ge1;
\end{equation*}
the Rational QEPF model is the case $\lambda=1$, with $\psi=(\mu,\gamma)$ and $q=q^{*}$ given by \eqref{eq:q_gamma_explicit}. Let
\begin{equation*}
q(u;\psi)=\frac{\partial Q(u;\psi)}{\partial u}, \qquad 0<u<1,
\end{equation*}
which is available in closed form from \eqref{eq:PQEPF_q}.

The proposed family is specified directly through its quantile function. Its distribution function $F=Q^{-1}$ has no closed form for general $(\gamma,\lambda)$, so an observation $x$ cannot be converted into its quantile level $u=F(x)$ by a formula. We therefore evaluate the likelihood in the quantile domain, following the general framework of \citet[Chapter~9]{Gilchrist2000}: for each trial value of the parameter, the levels $u_i$ are obtained by solving $Q(u_i;\psi)=x_i$ numerically. Because $Q(\cdot;\psi)$ is strictly increasing and continuous, this equation has a unique root whenever $x_i$ lies in the support, and a bracketing method such as bisection or Brent's method on $(0,1)$ finds it reliably.

\noindent
For an observation $x$ from a continuous distribution with quantile function $Q(\cdot;\psi)$, define the corresponding quantile level $u$ implicitly by
\begin{equation*}
Q(u;\psi)=x, \qquad 0<u<1.
\end{equation*}
Using the identity
\begin{equation*}
f(Q(u;\psi);\psi)=\frac{1}{q(u;\psi)},
\end{equation*}
the likelihood can be expressed in terms of the implied quantile levels.

\noindent
Let $x_1,\ldots,x_n$ be fully observed, and let $u_i=u_i(\psi)$ satisfy
\begin{equation*}
Q(u_i;\psi)=x_i, \qquad i=1,\ldots,n.
\end{equation*}
Then
\begin{equation*}
L(\psi)=\prod_{i=1}^n \frac{1}{q(u_i;\psi)},
\end{equation*}
so the log-likelihood is
\begin{equation}
\ell(\psi)
= -\sum_{i=1}^n \log q(u_i;\psi).
\label{eq:loglik_complete}
\end{equation}

\noindent
For survival data, let $(t_i,\delta_i)$ denote the observed time and censoring indicator, where $\delta_i=1$ for an event and $\delta_i=0$ for a right-censored observation, and let $Q(u_i;\psi)=t_i$. Since $S(t_i;\psi)=1-u_i$,
the log-likelihood under non-informative censoring is
\begin{equation}
\ell(\psi)
=
\sum_{i=1}^n
\left[
-\delta_i \log q(u_i;\psi)
+
(1-\delta_i)\log(1-u_i)
\right].
\label{eq:loglik_censored}
\end{equation}

\paragraph*{Score and standard errors.}
Differentiating $Q(u_i(\psi);\psi)=x_i$ gives the implicit derivative
\begin{equation}
\label{eq:implicit_grad}
\frac{\partial u_i}{\partial\psi}=-\frac{\partial Q(u;\psi)/\partial\psi}{q(u;\psi)}\bigg|_{u=u_i},
\end{equation}
so the score of \eqref{eq:loglik_complete} or \eqref{eq:loglik_censored} is available analytically once the $u_i$ are computed. Standard errors are obtained from the inverse of the observed information matrix, computed by differentiating the score numerically or analytically. To enforce $\mu>0$ and $\gamma>-1$ during optimization, we use $\mu=\exp(\vartheta_1)$ and $\gamma=\exp(\vartheta_2)-1$.

\paragraph*{Support and regularity.}
When $\lambda>1$, the support is $(0,\infty)$ and every observation is feasible. When $\lambda=1$, the upper endpoint $\omega(\psi)=\mu(1+\gamma)^{1/\gamma}$ depends on the parameters, so any $\psi$ with $\omega(\psi)\le\max_i t_i$ is infeasible and is assigned $\ell(\psi)=-\infty$. Since the density is positive at $\omega$, the Rational QEPF model is a non-regular problem in the sense of \citet{Smith1985}. Standard asymptotics apply when the MLE satisfies $\hat\omega>\max_i t_i$ by a clear margin, as in the heavily censored data of \Cref{sec:data-analysis}, but they may fail when $\hat\omega$ approaches the largest observation.

\paragraph*{Testing for bounded support.}
The Rational QEPF family corresponds to $H_0\colon\lambda=1$, which lies on the boundary of the parameter space $\lambda\ge1$. Let $\Lambda=2\{\ell(\hat\mu,\hat\gamma,\hat\lambda)-\ell(\tilde\mu,\tilde\gamma,1)\}$ be the likelihood ratio statistic, where tildes denote the Rational QEPF fit. Under the conditions of \citet{SelfLiang1987}, $\Lambda$ is asymptotically distributed as the mixture $\tfrac12\chi^2_0+\tfrac12\chi^2_1$, so the $p$-value is $\tfrac12\Pr(\chi^2_1>\Lambda)$ for $\Lambda>0$, and the 5\% critical value is $2.706$. Because the null model is itself non-regular, we check the size of this test by simulation in \Cref{sec:simulation}. We parameterize $\lambda=1+\vartheta_3^2$ during optimization, so the boundary value is attainable.

\paragraph*{Restricted QEPF.}
With right-censored data the survival function $S=1-F$ is identified only up to the largest follow-up time $\tau$, so the tail mean in \eqref{eq:qepf_def}, and hence $P(u)$, cannot be estimated nonparametrically when $Q(u)$ is close to $\tau$. Following the horizon restriction suggested by \citet[Sec.~3.1]{qepf2026}, and in the manner of the restricted mean survival time \citep{RoystonParmar2013}, we define the restricted QEPF
\begin{equation}
\label{eq:restricted_qepf}
P_\tau(u)=1+\frac{1}{S(Q(u))\,Q(u)}\int_{Q(u)}^{\tau}S(t)\,dt,\qquad 0<u<F(\tau).
\end{equation}
For a continuous distribution $S(Q(u))=1-u$, so $P_\tau(u)$ is obtained from $P(u)=1+M(u)/Q(u)$ by replacing the mean residual life beyond $Q(u)$ with its value truncated at $\tau$. Hence $1\le P_\tau(u)\le P(u)$, and $P_\tau(u)\to P(u)$ as $\tau\to\infty$. A nonparametric estimate of $P_\tau(u)$ is obtained by replacing $S$ and $Q$ with the Kaplan--Meier estimate \citep{kaplanmeier1958} and its quantile function. For a fitted model, $S$ and $Q$ are those of the model, so that the two estimates are compared at the same $\tau$.

\noindent
Competing parametric models are compared using the log-likelihood, AIC \citep{Akaike1974} and BIC \citep{Schwarz1978}.

\section{Simulation study}
\label{sec:simulation}

We assess the finite-sample behaviour of the maximum likelihood estimator and of the boundary test of \Cref{sec:inference}. Samples are generated as $X=Q(U;\psi)$ with $U$ uniform on $(0,1)$, under three settings with $\mu=1$:
\begin{equation*}
(\gamma,\lambda)=(0.5,\,1),\qquad(\gamma,\lambda)=(0.5,\,1.5),\qquad(\gamma,\lambda)=(-0.5,\,1.3).
\end{equation*}
The first is the Rational QEPF family. The other two have Pareto-type tails with tail indices $3$ and $4.33$, so the variance is finite in both. Right-censoring times are drawn independently from an exponential distribution whose rate gives a censoring proportion of $0\%$ or $30\%$. The sample sizes are $n=100,250,500$, with $R=300$ replications for each combination. Each fit uses the Nelder--Mead method with three starting values and the parameterization of \Cref{sec:inference}. Standard errors come from the observed information, and the nominal level of the Wald intervals is $95\%$.

\paragraph*{Size of the boundary test.}
\Cref{tab:sim_size} reports, under $H_0\colon\lambda=1$, the proportion of samples with $\hat\lambda=1$ and the rejection rate of the likelihood ratio test at the nominal $5\%$ level. The estimate lies on the boundary in $89$--$95\%$ of the samples, far more than the $50\%$ predicted by the regular theory of \citet{SelfLiang1987}. This reflects the non-regularity of the Rational QEPF model: a sample with a sharp upper endpoint is poorly described by a Pareto-type tail, so the extended fit rarely moves away from $\lambda=1$. As a result, the $\tfrac12\chi^2_0+\tfrac12\chi^2_1$ calibration is conservative, with empirical size between $0.3\%$ and $1.3\%$. The test is therefore valid but conservative. When an exact level is needed, a parametric bootstrap \citep{EfronTibshirani1993} from the fitted Rational QEPF model can be used to calibrate it.

\begin{table}[ht]
\centering
\small
\caption{Behaviour of the boundary test under $H_0\colon\lambda=1$, with $(\mu,\gamma)=(1,0.5)$ and $R=300$ replications. \protect\sagenew{Source: simulation study, Section~\ref{sec:simulation}.}}
\label{tab:sim_size}
\begin{tabular}{lrcc}
\toprule
Censoring & $n$ & $\Pr(\hat\lambda=1)$ & Rejection rate (nominal $0.05$) \\
\midrule
0\% & 100 & $0.93$ & $0.013$ \\
0\% & 250 & $0.92$ & $0.007$ \\
0\% & 500 & $0.93$ & $0.003$ \\
30\% & 100 & $0.89$ & $0.010$ \\
30\% & 250 & $0.93$ & $0.003$ \\
30\% & 500 & $0.95$ & $0.003$ \\
\bottomrule
\end{tabular}
\end{table}

\paragraph*{Estimation under a heavy tail.}
\Cref{tab:sim_est} summarizes the two settings with $\lambda>1$. Because a few censored samples of size $100$ produce very large estimates of $\lambda$ and $\mu$, we report the median bias, and we compute the root mean squared error after excluding fits with $\hat\lambda>5$ or $\hat\mu>5$. Such fits occur in $2.3\%$ and $0.3\%$ of the samples for $(0.5,1.5)$ with $30\%$ censoring and $n=100,250$, and in $12.7\%$, $4.3\%$ and $0.7\%$ of the samples for $(-0.5,1.3)$ with $30\%$ censoring and $n=100,250,500$. They do not occur in any other setting. A Wald interval is not available when $\hat\lambda=1$, and such samples are counted as not covering the true value.

\begin{table}[ht]
\centering
\footnotesize
\setlength{\tabcolsep}{3.5pt}
\caption{Median bias, root mean squared error (RMSE) and coverage of $95\%$ Wald intervals for the extended Rational QEPF model with $\mu=1$; $\Pr(\hat\lambda=1)$ is the proportion of samples with the estimate on the boundary, and power is the rejection rate of $H_0\colon\lambda=1$ at the $5\%$ level. $R=300$ replications. \protect\sagenew{Source: simulation study, Section~\ref{sec:simulation}.}}
\label{tab:sim_est}
\begin{tabular}{@{}lcrrrrrrrrrrcc@{}}
\toprule
& & & \multicolumn{3}{c}{Median bias} & \multicolumn{3}{c}{RMSE} & \multicolumn{3}{c}{Coverage} & & \\
\cmidrule(lr){4-6}\cmidrule(lr){7-9}\cmidrule(lr){10-12}
$(\gamma,\lambda)$ & Cens. & $n$ & $\hat\mu$ & $\hat\gamma$ & $\hat\lambda$ & $\hat\mu$ & $\hat\gamma$ & $\hat\lambda$ & $\hat\mu$ & $\hat\gamma$ & $\hat\lambda$ & $\Pr(\hat\lambda=1)$ & Power \\
\midrule
$(0.5,\,1.5)$ & 0\% & 100 & $-0.011$ & $-0.011$ & $-0.054$ & $0.114$ & $0.577$ & $0.253$ & $0.94$ & $0.92$ & $0.87$ & $0.01$ & $0.95$ \\
 & 0\% & 250 & $0.002$ & $0.009$ & $-0.015$ & $0.076$ & $0.331$ & $0.172$ & $0.94$ & $0.94$ & $0.92$ & $0.00$ & $1.00$ \\
 & 0\% & 500 & $-0.006$ & $-0.015$ & $-0.013$ & $0.050$ & $0.224$ & $0.099$ & $0.95$ & $0.94$ & $0.93$ & $0.00$ & $1.00$ \\
 & 30\% & 100 & $0.016$ & $-0.040$ & $-0.076$ & $0.269$ & $0.782$ & $0.646$ & $0.79$ & $0.83$ & $0.75$ & $0.14$ & $0.58$ \\
 & 30\% & 250 & $0.001$ & $-0.005$ & $-0.032$ & $0.110$ & $0.444$ & $0.311$ & $0.93$ & $0.92$ & $0.88$ & $0.01$ & $0.94$ \\
 & 30\% & 500 & $-0.017$ & $-0.004$ & $-0.028$ & $0.074$ & $0.291$ & $0.203$ & $0.94$ & $0.92$ & $0.89$ & $0.00$ & $1.00$ \\
$(-0.5,\,1.3)$ & 0\% & 100 & $0.010$ & $-0.020$ & $-0.075$ & $0.144$ & $0.209$ & $0.289$ & $0.71$ & $0.74$ & $0.70$ & $0.25$ & $0.42$ \\
 & 0\% & 250 & $0.008$ & $-0.016$ & $-0.021$ & $0.089$ & $0.108$ & $0.154$ & $0.92$ & $0.88$ & $0.83$ & $0.02$ & $0.90$ \\
 & 0\% & 500 & $0.005$ & $-0.009$ & $-0.022$ & $0.060$ & $0.079$ & $0.115$ & $0.95$ & $0.92$ & $0.87$ & $0.00$ & $1.00$ \\
 & 30\% & 100 & $-0.027$ & $0.023$ & $-0.300$ & $0.277$ & $0.285$ & $0.643$ & $0.44$ & $0.46$ & $0.46$ & $0.52$ & $0.12$ \\
 & 30\% & 250 & $0.004$ & $0.002$ & $-0.062$ & $0.281$ & $0.184$ & $0.669$ & $0.62$ & $0.66$ & $0.66$ & $0.32$ & $0.23$ \\
 & 30\% & 500 & $0.004$ & $-0.006$ & $-0.024$ & $0.145$ & $0.128$ & $0.372$ & $0.76$ & $0.79$ & $0.74$ & $0.20$ & $0.50$ \\
\bottomrule
\end{tabular}
\end{table}

Without censoring, the median bias is small throughout, the RMSE decreases at roughly the $n^{-1/2}$ rate, and for $n\ge250$ the coverage is close to the nominal level for $\mu$ and $\gamma$ and between $0.83$ and $0.93$ for $\lambda$. At $n=100$ the coverage of the interval for $\lambda$ is lower, mainly because in $1\%$ and $25\%$ of the samples, respectively, the estimate falls on the boundary. With $30\%$ censoring, the setting $(0.5,1.5)$ behaves well from $n=250$ onwards. The setting $(-0.5,1.3)$ is the most difficult: its tail is light ($\lambda$ close to $1$), and censoring removes much of the information about it. Even at $n=500$, $20\%$ of the samples give $\hat\lambda=1$ and the power of the test is $0.50$. When $\hat\lambda>1$, however, the Wald intervals in this setting have coverage between $0.92$ and $1.00$. In practice, therefore, a non-rejection of $\lambda=1$ for heavily censored data indicates that the data cannot distinguish a bounded tail from a moderately heavy one, rather than evidence for bounded support. This is the situation of the PBC data in \Cref{sec:data-analysis}.

\section{Data analysis}
\label{sec:data-analysis}

\subsection{Primary biliary cirrhosis data}
\label{subsec:pbc}

We illustrate the proposed model using the \emph{Cirrhosis Patient Survival
Prediction} dataset from the UCI Machine Learning Repository, derived from the
Mayo Clinic primary biliary cirrhosis (PBC) study conducted between 1974 and
1984 \citep{UCI_Cirrhosis,Dickson1989PBC}. The dataset records follow-up time
(in days) from study entry to the earliest of death, liver transplantation, or
administrative end of follow-up. For details of the original study and data
collection, see \citep{Dickson1989PBC,MayoPBC}.

The dataset contains 418 patients in total, including 312 randomized trial
participants and 106 additional non-randomized patients followed for survival.
In the full cohort, there were 161 deaths, 25 liver transplantations, and 232
other censored observations. Among the randomized patients, 158 received
D-penicillamine and 154 received placebo. 

The variable \texttt{N\_Days} denotes follow-up time, and the variable
\texttt{Status} takes values \texttt{D} for death, \texttt{C} for censoring,
and \texttt{CL} for censoring due to liver transplantation. Death is treated as
the event of interest, while both \texttt{C} and \texttt{CL} are treated as
right censoring. Thus, the analysis is a cause-specific analysis of time to
death.

Analyses are carried out for the overall cohort and separately by randomized
treatment assignment (\texttt{Drug}: D-penicillamine versus placebo).

\paragraph*{Models and estimation.}
The Rational QEPF model and its extension of \Cref{sec:extension} are fitted under right censoring by maximum
likelihood using \eqref{eq:loglik_censored}, with likelihood
evaluation based on numerical inversion of the quantile function. For comparison, standard parametric survival models are also fitted using intercept-only specifications: Weibull, lognormal, and
loglogistic models via \texttt{survreg} in the \texttt{survival} package
\citep{TherneauGrambsch2000}, and Gamma and generalized Gamma models via
\texttt{flexsurv} \citep{Jackson2016Flexsurv}. Models are compared using
log-likelihood, Akaike information criterion (AIC), and Bayesian information
criterion (BIC) \citep{Akaike1974,Schwarz1978}.

\paragraph*{Overall and treatment-stratified model comparison}

\begin{table}[ht]
\centering
\small
\caption{Model comparison for the PBC dataset overall and by
treatment arm. \protect\sagenew{Source: Mayo Clinic PBC data \citep{UCI_Cirrhosis,Dickson1989PBC}.}}
\label{tab:pbc_modelcomp_overall_drug}
\begin{tabular}{llrrrr}
\toprule
Cohort & Model & logLik & $k$ & AIC & BIC \\
\midrule
\multicolumn{6}{l}{\textbf{Overall}}\\
\quad Rational QEPF     &  & $-1530.4646$ & 2 & $3064.9291$ & $3073.0001$ \\
\quad Extended Rational QEPF &  & $-1530.4646$ & 3 & $3066.9291$ & $3079.0356$ \\
\quad Weibull          &  & $-1531.0174$ & 2 & $3066.0349$ & $3074.1058$ \\
\quad Gamma            &  & $-1531.0736$ & 2 & $3066.1471$ & $3074.2181$ \\
\quad Gen.\ Gamma      &  & $-1530.8413$ & 3 & $3067.6825$ & $3079.7890$ \\
\quad Loglogistic      &  & $-1532.7886$ & 2 & $3069.5773$ & $3077.6482$ \\
\quad Lognormal        &  & $-1535.8606$ & 2 & $3075.7211$ & $3083.7921$ \\
\addlinespace
\multicolumn{6}{l}{\textbf{D-penicillamine}}\\
\quad Rational QEPF     &  & $-615.8941$  & 2 & $1235.7883$ & $1241.9135$ \\
\quad Extended Rational QEPF &  & $-615.6595$  & 3 & $1237.3191$ & $1246.5069$ \\
\quad Weibull          &  & $-615.7392$  & 2 & $1235.4784$ & $1241.6036$ \\
\quad Gamma            &  & $-615.8766$  & 2 & $1235.7531$ & $1241.8783$ \\
\quad Gen.\ Gamma      &  & $-615.6272$  & 3 & $1237.2543$ & $1246.4421$ \\
\quad Loglogistic      &  & $-616.6065$  & 2 & $1237.2130$ & $1243.3382$ \\
\quad Lognormal        &  & $-619.0836$  & 2 & $1242.1673$ & $1248.2925$ \\
\addlinespace
\multicolumn{6}{l}{\textbf{Placebo}}\\
\quad Rational QEPF     &  & $-571.8792$  & 2 & $1147.7584$ & $1153.8323$ \\
\quad Extended Rational QEPF &  & $-571.8792$  & 3 & $1149.7584$ & $1158.8692$ \\
\quad Weibull          &  & $-572.4588$  & 2 & $1148.9176$ & $1154.9915$ \\
\quad Gamma            &  & $-572.4712$  & 2 & $1148.9424$ & $1155.0163$ \\
\quad Gen.\ Gamma      &  & $-571.8973$  & 3 & $1149.7945$ & $1158.9054$ \\
\quad Loglogistic      &  & $-573.3288$  & 2 & $1150.6576$ & $1156.7315$ \\
\quad Lognormal        &  & $-573.6784$  & 2 & $1151.3569$ & $1157.4308$ \\
\bottomrule
\end{tabular}
\end{table}

\Cref{tab:pbc_modelcomp_overall_drug} shows that the Rational QEPF
model is competitive with standard parametric alternatives both overall and
within each treatment arm. It gives the smallest AIC and BIC for the overall
cohort and for the placebo arm, while in the D-penicillamine arm the Weibull
model is slightly preferred. All AIC differences among the two-parameter models other than the lognormal are below $2$, so the data do not discriminate between them, and the conclusion is one of comparable fit rather than dominance.

For the overall cohort and the placebo arm, the extended Rational QEPF fit lies exactly on the boundary $\hat\lambda=1$, so the likelihood ratio statistic for $H_0\colon\lambda=1$ is $\Lambda=0$ ($p=0.5$). In the D-penicillamine arm, $\hat\lambda=1.95$ but $\Lambda=0.47$ ($p=0.25$). Bounded support is therefore not rejected in any of the three analyses. This is to be expected: about 60\% of the observations are censored, and the Kaplan--Meier curves end at survival probabilities between $0.32$ and $0.36$, so the data carry little information about the upper tail.

For the Rational QEPF model, the fitted parameters (standard errors in parentheses) are
\begin{equation*}
\begin{aligned}
\text{Overall: }\ &\widehat{\mu}=3519.5\ (314.2),\ \widehat{\gamma}=0.549\ (0.289),\\
\text{D-penicillamine: }\ &\widehat{\mu}=3260.7\ (320.8),\ \widehat{\gamma}=0.930\ (0.468),\\
\text{Placebo: }\ &\widehat{\mu}=3624.5\ (553.4),\ \widehat{\gamma}=0.523\ (0.494).
\end{aligned}
\end{equation*}
In all three fits $\widehat{\gamma}>1/2$, so the fitted
hazard--quantile function $H^{*}(u)$ is strictly increasing on $(0,1)$
(see \Cref{subsec:family_summary}). However, the standard errors show that $\gamma<1/2$ cannot be excluded, so the shape of $H^{*}$ is not established by these data. The fitted upper endpoint for the overall cohort is $\hat\omega=7810$ days, against a maximum follow-up of $4795$ days. Statements about survival beyond follow-up, and $\hat\mu$ itself, therefore rest on extrapolation by the model.

Within the observation range, the fitted Rational QEPF, Weibull and Gamma survival curves for the two treatment arms are almost indistinguishable from one another (not shown), in agreement with the likelihood-based comparison in \Cref{tab:pbc_modelcomp_overall_drug}.

\paragraph*{Restricted QEPF.} In the two arms the Kaplan--Meier estimate ends at $\hat u_{\max}=1-\hat S(\tau)=0.681$ (D-penicillamine, $\tau=4556$ days) and $0.639$ (placebo, $\tau=4523$ days), so $P(u)$ itself cannot be estimated from the data beyond these levels. We therefore compare the models through the restricted QEPF \eqref{eq:restricted_qepf}, with $\tau$ the last follow-up time in each arm.
 \Cref{fig:pbc_persistence} shows the Kaplan--Meier estimate of $P_\tau(u)$ for $0.10\le u\le \hat u_{\max}-0.02$, with a pointwise $95\%$ bootstrap band \citep{EfronTibshirani1993} ($500$ resamples within each arm), together with the restricted QEPFs implied by the fitted Rational QEPF and Weibull models. The Rational QEPF curve lies inside the band at $96\%$ of the grid points for D-penicillamine and $90\%$ for placebo; the corresponding figures for the Weibull model are $93\%$ and $85\%$. Both models deviate most at small $u$, where the band is widest. The diagnostic thus supports the Rational QEPF fit but, like \Cref{tab:pbc_modelcomp_overall_drug}, does not separate it from the Weibull model. In the upper grid points some bootstrap resamples end before $u$, so the band there is narrower than its nominal level.

\begin{figure}[ht]
\centering
\includegraphics[width=\textwidth]{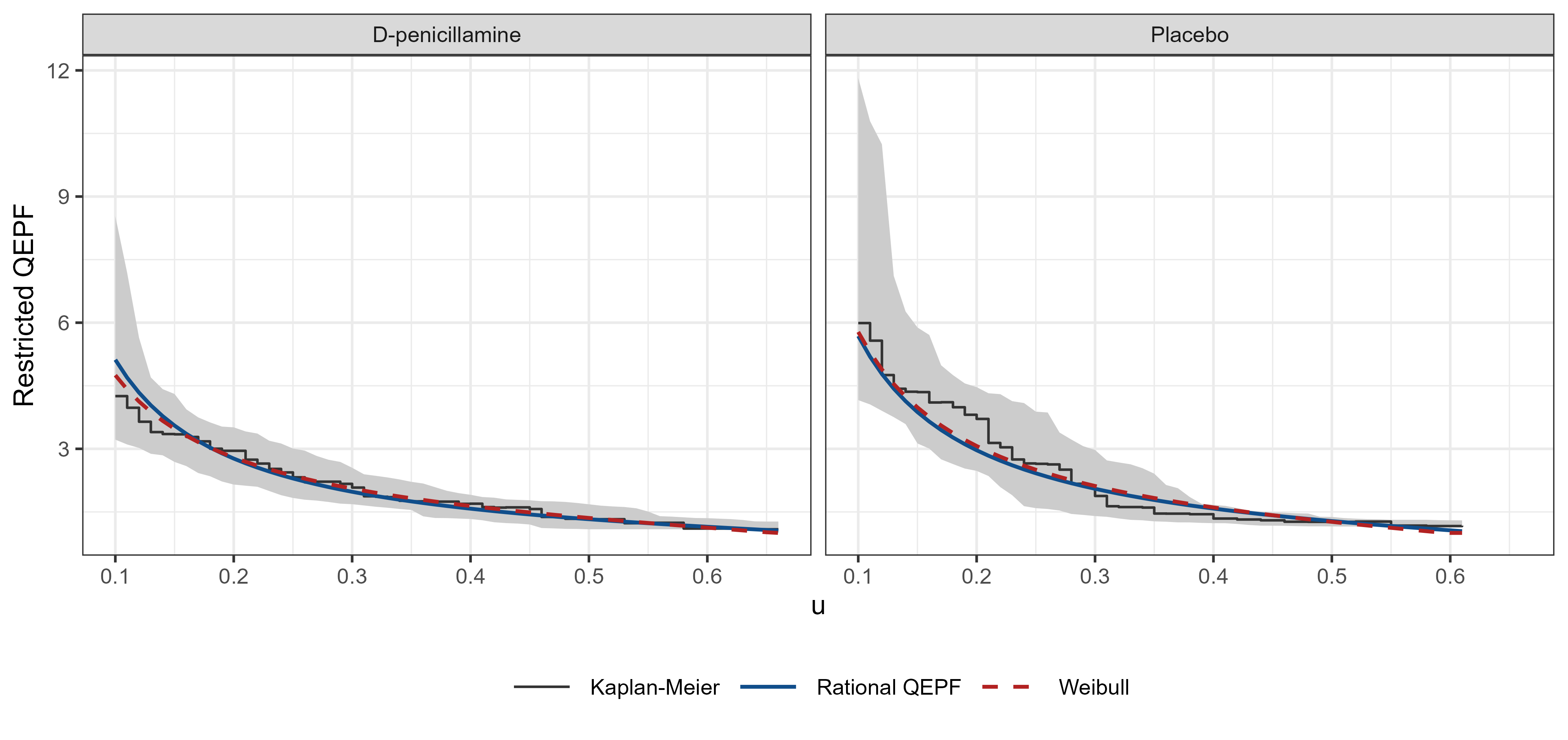}
\caption{Restricted QEPF $P_\tau(u)$ for the randomized PBC patients by treatment arm: Kaplan--Meier estimate with pointwise $95\%$ bootstrap band (grey), and the curves implied by the fitted Rational QEPF (solid) and Weibull (dashed) models. \protect\sagenew{Source: Mayo Clinic PBC data \citep{UCI_Cirrhosis,Dickson1989PBC}.}}
\label{fig:pbc_persistence}
\end{figure}

\paragraph*{Comparing treatment arms.} In the Rational QEPF family, \eqref{eq:canonical_P} gives, for a treatment arm $T$ and a control arm $C$,
\begin{equation*}
\frac{M_T(u)/Q_T(u)}{M_C(u)/Q_C(u)}=\frac{1+\gamma_C}{1+\gamma_T}=:\theta,\qquad 0<u<1,
\end{equation*}
which does not depend on $u$. Among patients surviving beyond the $u$-th quantile, the expected residual survival relative to that quantile in arm $T$ is $\theta$ times that in arm $C$, at every quantile level. Unlike the hazard ratio, this contrast does not require proportional hazards: two arms with $\gamma_T\neq\gamma_C$ have non-proportional hazard--quantile functions, yet $\theta$ is constant. In this respect $\theta$ is analogous to the odds ratio of the proportional odds model \citep{KirmaniGupta2001}, which is constant over time; $\theta$ is constant over quantile levels. A two-arm comparison is then summarized by $(\mu_T/\mu_C,\theta)$, and $H_0\colon\theta=1$ ($\gamma_T=\gamma_C$, means free) is tested by a likelihood ratio test with one degree of freedom. The constancy of $\theta$ holds only for $\lambda=1$, so $\theta$ is reported only when the test of \Cref{sec:inference} does not reject $\lambda=1$ in either arm.

For the randomized PBC patients, $\hat\theta=0.79$ with $95\%$ confidence interval $(0.36,\,1.75)$, obtained by the delta method \citep{Serfling1980} on the log scale. The likelihood ratio statistic for $H_0\colon\gamma_T=\gamma_C$ is $0.35$ ($p=0.55$; common $\hat\gamma=0.752$), consistent with the known absence of a D-penicillamine effect on survival. The tests of $\lambda=1$ give $p=0.25$ (D-penicillamine) and $p=0.50$ (placebo); with heavy censoring, these tests have low power. The $90\%$ interval for $\theta$, $(0.41,\,1.54)$, is too wide to establish equivalence within the margin $[0.8,1.25]$, which reflects the heavy censoring. In biosimilar comparisons, which motivated the QEPF \citep{qepf2026}, the same interval gives two one-sided tests of equivalence in $\theta$.

The fitted QEPF at the $90$th percentile is $\hat P(0.9)=1.058\ (1.030,\,1.085)$ for D-penicillamine and $1.073\ (1.027,\,1.119)$ for placebo, with difference $-0.015\ (-0.069,\,0.038)$. Because $u=0.9$ lies beyond $\hat u_{\max}$ in both arms, these values are extrapolations by the fitted model and are not checked by the data. The intervals are short only because $P(0.9)=(1+0.9\gamma)/\{0.9(1+\gamma)\}$ varies little with $\gamma$.

\subsection{Advanced lung cancer data}
\label{subsec:lung}

The second dataset is the North Central Cancer Treatment Group study of patients with advanced lung cancer \citep{Loprinzi1994}, available as \texttt{lung} in the \texttt{survival} package. It records survival time in days for $n=228$ patients, of whom 165 died (28\% censored), with a maximum follow-up of 1,022 days. In contrast to PBC, most patients die during follow-up, so the data are informative about the upper tail.

\begin{table}[ht]
\centering
\small
\caption{Model comparison for the lung cancer data ($n=228$, 165 deaths). \protect\sagenew{Source: data from \citet{Loprinzi1994}.}}
\label{tab:lung_modelcomp}
\begin{tabular}{lrrrr}
\toprule
Model & logLik & $k$ & AIC & BIC \\
\midrule
Rational QEPF       & $-1163.7198$ & 2 & $2331.4396$ & $2338.2983$ \\
Extended Rational QEPF  & $-1153.9089$ & 3 & $2313.8178$ & $2324.1058$ \\
Weibull            & $-1153.8512$ & 2 & $2311.7024$ & $2318.5611$ \\
Gamma              & $-1154.7346$ & 2 & $2313.4693$ & $2320.3280$ \\
Gen.\ Gamma        & $-1153.6898$ & 3 & $2313.3796$ & $2323.6676$ \\
Loglogistic        & $-1160.9306$ & 2 & $2325.8612$ & $2332.7199$ \\
Lognormal          & $-1169.2691$ & 2 & $2342.5381$ & $2349.3968$ \\
\bottomrule
\end{tabular}
\end{table}

\Cref{tab:lung_modelcomp} shows that the Rational QEPF model fits these data poorly. Its fitted endpoint is $\hat\omega=1067$ days, barely above the largest observation, which is the non-regular situation described in \Cref{sec:inference}. The extended Rational QEPF model gives $\hat\mu=396.4\ (33.1)$, $\hat\gamma=1.73\ (0.75)$ and $\hat\lambda=1.43\ (0.22)$. The likelihood ratio statistic for $H_0\colon\lambda=1$ is $\Lambda=19.6$ ($p<10^{-5}$), so bounded support is rejected. With $\hat\lambda=1.43$, the fitted tail index is $\hat\lambda/(\hat\lambda-1)=3.3$, so the variance is finite. The interpretation is direct: for survival times beyond a high quantile, the mean survival is about $1.43$ times that quantile. In terms of fit, the extended Rational QEPF model is equivalent to the gamma and generalized gamma models (AIC within $0.5$), while the Weibull model is the most parsimonious description of these data.

\section{Conclusion}
\label{sec:conclusion}

We have shown that the quantile-based effectiveness persistence function determines a distribution up to scale, and we have characterized the functions that can be a QEPF. Under the boundary conditions $Q(0^+)=0$ and $Q(1^-)<\infty$, every rational specification of the QEPF reduces to the canonical form $(1+\gamma u)/\{(1+\gamma)u\}$, in which the relative mean residual quantile function is proportional to the survival odds. The canonical form defines the two-parameter Rational QEPF family, indexed by the mean $\mu$ and the shape parameter $\gamma$, whose quantile function, tail mean, mean residual quantile function, hazard--quantile function and L-moments are available in closed form.

Relaxing the upper boundary condition multiplies the Rational QEPF by a constant $\lambda=P(1^-)\ge1$ and gives Pareto-type tails for $\lambda>1$. The Rational QEPF family is the boundary case $\lambda=1$, so the question of bounded support becomes a boundary likelihood ratio test.

In the heavily censored PBC data, the Rational QEPF model is not rejected and fits as well as the standard two-parameter models; its finite endpoint, however, lies well beyond the follow-up and should be read as an extrapolation. In the lung cancer data, bounded support is rejected, and the extended model fits as well as the gamma and generalized gamma models, with a directly interpretable tail parameter. For two treatment arms, the Rational QEPF family gives a contrast $\theta=(1+\gamma_C)/(1+\gamma_T)$ that is constant over quantile levels without assuming proportional hazards; in the PBC trial it gives no evidence of a difference between the arms. The simulation study shows that maximum likelihood estimation works well for moderate samples, that heavy censoring weakens the information about $\lambda$, and that the boundary test is conservative when the null model is bounded.

Several extensions remain of interest. These include regression models in which $\mu$ or $\lambda$ depends on covariates, which would allow treatment effects to be expressed as changes in $\lambda$; and formal goodness-of-fit tests based on the nonparametric QEPF that account for censoring beyond the end of follow-up.

\begingroup\sagecolor
\section*{Funding}
The author(s) disclosed receipt of the following financial support for the research, authorship, and/or publication of this article: This work was supported by the Council of Scientific and Industrial Research (CSIR), Government of India, under the Emeritus Scientist Scheme (File No.\ 21/1191/25/EMR-II).

\section*{Declaration of conflicting interests}
The authors declare that there is no conflict of interest.

\section*{Ethical approval}
Not applicable. All data used are publicly available and de-identified.

\section*{Data availability statement}
The primary biliary cirrhosis and lung cancer data are available as \texttt{pbc} and \texttt{lung} in the R package \texttt{survival}. The blood cancer data are given in \citet[Table~4]{Abouammoh1994}. The R code is available from the corresponding author on request.
\endgroup

\section*{Author contributions}
All authors contributed to the conception of the work, reviewed the manuscript, and approved the final version.

\bibliographystyle{unsrtnat}
\bibliography{references}
\end{document}